\documentclass[10 pt]{scrartcl}

\usepackage[utf8]{inputenc}
\usepackage[english]{babel}
\usepackage{pdfpages}
\usepackage{amsmath,amssymb}
\usepackage{graphicx}
\usepackage{float}
\usepackage{cite}
\usepackage{subfigure}
\usepackage{amsthm}
\usepackage{hyperref}
\usepackage{cleveref}
\usepackage{geometry}

\newtheorem{theorem}{Theorem}
\newtheorem{corollary}{Corollary}
\newtheorem{definition}{Definition}

\usepackage[mathlines]{lineno}
\newcommand*\patchAmsMathEnvironmentForLineno[1]{%
	\expandafter\let\csname old#1\expandafter\endcsname\csname #1\endcsname
	\expandafter\let\csname oldend#1\expandafter\endcsname\csname end#1\endcsname
	\renewenvironment{#1}%
	{\linenomath\csname old#1\endcsname}%
	{\csname oldend#1\endcsname\endlinenomath}}%
\newcommand*\patchBothAmsMathEnvironmentsForLineno[1]{%
	\patchAmsMathEnvironmentForLineno{#1}%
	\patchAmsMathEnvironmentForLineno{#1*}}%
\AtBeginDocument{%
	\patchBothAmsMathEnvironmentsForLineno{equation}%
	\patchBothAmsMathEnvironmentsForLineno{align}%
	\patchBothAmsMathEnvironmentsForLineno{flalign}%
	\patchBothAmsMathEnvironmentsForLineno{alignat}%
	\patchBothAmsMathEnvironmentsForLineno{gather}%
	\patchBothAmsMathEnvironmentsForLineno{multline}%
}

\title{\LARGE\textbf{Connected Assembly and Reconfiguration\\ by Finite Automata}}

\author{\normalsize Sándor P. Fekete$^*$, Eike Niehs$^*$, Christian Scheffer$^*$, Arne Schmidt\thanks{Department of Computer Science, TU Braunschweig, Germany.
		$\{$s.fekete, e.niehs, c.scheffer, arne.schmidt$\}$@tu-bs.de}%
}
\date{}

\begin{document}
	\maketitle
\begin{abstract}
	We consider methods for connected reconfigurations by finite automate
in the so-called \emph{hybrid} or \emph{Robot-on-Tiles} model of
programmable matter, in which a number of simple robots move on and
rearrange an arrangement of passive tiles in the plane that form \emph{polyomino} shapes,
making use of a supply of additional tiles that can be placed.
We investigate the problem of reconfiguration under the constraint of maintaining connectivity of the tile arrangement;
this reflects scenarios in which disconnected subarrangements may drift 
apart, e.g., in the absence of gravity in space. We show that two finite automata
suffice to mark a bounding box, which can then be used as a stepping stone for
more complex operations, such as scaling a tile arrangement by a given factor,
rotating arrangements, or copying arrangements to a different location.
We also describe an algorithm for scaling monotone polyominoes without the help
of a bounding box.
\end{abstract}

\section{Introduction}

The field of programmable matter deals with materials or structures that modify or adjust their
physical properties like shape, density or conductivity in complex manners.
These changes can be carried out by the material itself, or be controlled from the outside.
That allows components to interact with each other or its environment, as well as carry out tasks
like self-assembling, sensing or storing states and carrying out computations.
The involved materials or structures are often biologically inspired~\cite{Schmied2017,McEvoy2015}; 
possible application scenarios can be found
in different research areas, like architecture \cite{Willmann2013,Tibbits2012},
aircraft \cite{Barbarino2011}, surgery \cite{Smith2011} or nano-robotics
\cite{Dolev2016,Dolev2018}.  Challenges arise from the serious constraints of individual particles,
such as limited memory, computing capabilities, communication or field of vision.
This makes it relevant to develop algorithms that only make minimal assumptions on the capabilities
of robots.


In this paper, we investigate the problem of recognizing and modifying a given
connected arrangement $P$ of unit square tiles forming a polyomino by a small
number of finite automata, while maintaining connectivity between all tiles and
all robots during the execution of an algorithm.  Given the limitations of
finite automata, these are incapable of making use of any complex information about $P$;
instead, they need to rely on the structure of $P$ itself for carrying out their task,
making use of suitable protocols that result in the desired outcomes, such as scaling a
$P$ by a constant factor $c$, producing a copy of $P$, or rotating $P$. 

\subsection{Our Results}

Our approach makes use of first computing a bounding box for the arrangement $P$,
which is then used for more advanced tasks. Our results include the following. 

\begin{enumerate}
	\item Given a polyomino \(P\), constructing a bounding box surrounding \(P\) can be done using two robots in \(O(max(w,h) \cdot (wh + k \cdot |\partial P|))\) steps (see Theorem~\ref{th:bounding_box}).
In case of simply connected (i.e., hole-free) polyominoes, we show that one robot is 
sufficient (see Corollary~\ref{cor:bounding_box_simple}).
	\item Given a polyomino \(P\) that is already surrounded by a bounding
box, scaling by a constant scaling factor \(c\) next to the initial position of
\(P\) can be done with two robots in \(O(wh \cdot (c^2 + cw + ch))\) steps (see
Theorem~\ref{th:scaling}).  For simply connected polyominoes, one robot is
sufficient (see Corollary~\ref{cor:scaling_simple}).
	\item For a polyomino \(P\) that is already scaled by a known scaling factor \(c\), 
	we can reverse the process by down-scaling $P$ by a factor of \(\frac{1}{c}\) in \(O(wh \cdot (c^2 + cw + ch))\) 
	steps using one robot (see Corollary~\ref{cor:down_scaling}).
	\item We show that with the help of a bounding box, other reconfiguration
algorithms can also be carried out in a connected fashion
(see Theorem~\ref{th:adapting}).
	\item Given a monotone polyomino \(P\), 
	scaling by a constant scaling factor \(c\) can be done in 
	\(O(w^2hc^2)\) steps using two robots and 
	without requiring a bounding box construction (see Theorem~\ref{th:scale_monotone}).
\end{enumerate}

\textbf{Overview:} In \Cref{Prel} we start by introducing the basic model and definitions.
Our algorithm for constructing a bounding box with the help of just two finite-state robots is presented in \Cref{BB}.
In \Cref{Scale} we present our algorithm for scaling a given polyomino when the bounding box is already constructed.
The scaling itself can be carried out with one robot; as the construction of a bounding box is required as a stepping stone, 
the overall protocol also requires two robots.
The method for down-scaling an enlarged polyomino is described in \Cref{DownScaling},
while the strategy for adapting other algorithms to the connectivity constraints
is presented in \Cref{AdaptAlg}.  Finally, we
describe an algorithm for scaling monotone polyominoes without the previous
bounding box construction in \Cref{monotone}.

\subsection{Related Work}

In 1954 Golomb published his work ''Checker Boards and Polyominoes''
\cite{Golomb1954}, which appears to be one of the first articles dealing with
polyominoes.  Since then, a considerable body of work has been published.

\textbf{Tiling} deals with the capability of polyominoes
to tile the plane or specific shaped regions.  A hierarchy of tiling
capabilities for polyominoes was established in \cite{Golomb1966}, including a
classification for all simpler polyominoes up through hexominoes.
Furthermore, there is work regarding bounded regions, like rectangles
\cite{Klarner1969,Reid1997}, or tiling with sets of different polyominoes
\cite{Golomb1970}.  More recent research focuses on parallel
\cite{Takefuji1990} or genetic algorithms \cite{Gwee1996},
three-dimensional issue~\cite{Yamamoto1998}, or deciding whether a given polyomino can tile
the plane or not~\cite{Beauquier1991}.

In the field of \textbf{molecular self-assembly}, Winfree~\cite{Winfree1998}
showed that the abstract Tile Assembly Model (aTAM) can perform universal
computation.  The aTAM uses Wang Tiles~\cite{Wang1961}, which have a square
shape and glue on every side.
Tiles with equal glues on a square side can stick together to form larger arrangements of tiles.
The aTAM and its extensions, e.g.,
\cite{Adleman1999,Aggarwal2005,Soloveichik2007,Demaine2008}, 
are studied in tile-based self-assembly.
Demaine et al.~\cite{Demaine2008} presented 
a generalization that uses staging and hence enables more flexibility for laboratory experiments.
Other interesting results within the staged self-assembly model are shown in~\cite{Demaine2017}.
Another generalization was defined by Fekete et al.~\cite{Fekete2014}, 
who use polyominoes consisting of any number of unit square tiles, instead of only single square tiles.
For a survey of self-assembly models and recent results, see Patitz~\cite{Patitz2014}.

Another active area of research considers \textbf{robots or agents on graphs}, which
is similar to our Robots-on-Tiles model.  Regarding mazes as a specialization
of graphs, Blum and Kozen~\cite{Blum1978} showed that two automata can search an arbitrary maze.
Other work considers general graph exploration (e.g.,
\cite{Panaite1999,Fraigniaud2005,Fleischer2005}), as a distributed or collaborative problem using multiple agents (e.g.,
\cite{Bender1994,Fraigniaud2006,Das2007,Brass2011}) or with space limitations (e.g.,
\cite{Fraigniaud2004,Fraigniaud2005,Diks2004,Gasieniec2007,Gasieniec2008}).

Another related family of problems arises from \textbf{rendezvous search}~\cite{Alpern1995}, 
where two robots try to meet in a known or unknown environment.
Anderson and Weber~\cite{Anderson1990} presented an optimal strategy for discrete locations.
The variant of deterministic rendezvous search was studied in~\cite{Dessmark2006,Marco2006,Ta-Shma2014}.
The generalization to more than two robots is known as the \emph{gathering problem}, which was 
investigated on graphs~\cite{Kamei2011,DiStefano2013,DAngelo2013} or in the plane~\cite{Flocchini2005,Cieliebak2003,Czyzowicz2009}.

Specific problems in the context of \textbf{programmable matter}
include self-folding matter \cite{Hawkes2010,Knaian2012,Felton2014} or
self-disassembling magnetic robot pebbles \cite{Gilpin2010}. Inspired
by single-celled amoeba, Derakhshandeh et al.~\cite{Derakhshandeh2014} introduced
a fundamental concept for studying algorithmic approaches for extremely simple robotic agents,
called the \emph{Amoebot model}, with a generalized variant and further recent results
presented in~\cite{Derakhshandeh2015a}.  The Amoebot model provides a framework based
on an equilateral triangular graph and active particles that can occupy a
single vertex or a pair of adjacent vertices within that graph.  With just a few
possibly movements, these particles can form different shapes like
lines, triangles or hexagons \cite{Derakhshandeh2015}, or perform leader election
as a stepping stone for further operations~\cite{Derakhshandeh2015a,Daymude2017}.
A universal shape formation algorithm in the Amoebot model was described by
Di Luna et al.~\cite{DiLuna2017}.  An algorithm for solving the problem of
coating an arbitrarily shaped object with a layer of self-organizing
programmable matter was presented by Derakhshandeh et al.~\cite{Derakhshandeh2016a} 
and analyzed in \cite{Derakhshandeh2016}.  Other models with active particles were introduced
in by Woods et al.~\cite{Woods2013} in the Nubot model and by Hurtado et al.~\cite{Hurtado2015} for modular
robots.  Gmyr et al.~\cite{Gmyr2017}  introduced a \emph{hybrid} model with two types of
particles: active robots acting like a deterministic finite automaton and
passive tile particles.  Furthermore, they presented algorithms for shape
formation \cite{Gmyr2018a} and shape recognition \cite{Gmyr2018} using robots
on tiles.  This resembles the "Robot-on-Tiles" model of our work, for which 
we introduced geometric algorithms for copying, reflecting,
rotating and scaling of a given polyomino as well as an algorithm for
constructing a bounding box surrounding a polyomino in \cite{Fekete2018}.

However, none of these models or methods preserve connectivity, which is the main constraint
addressed by this paper.

\section{Preliminaries}
\label{Prel}

In the following, we introduce the basic model used for our work,
including general definitions and a description of the applicable constraints.

\subsection{Model}

We consider an infinite \textit{square grid graph} \(G\), where
\(\mathbb{Z}^2\) defines the \textit{vertices}; for every two vertices with
distance one there is a corresponding \textit{edge} in \(G\).
We use the compass directions \((N,E,S,W)\) for orientation when moving on the grid and may use \textit{up, right, down} and \textit{left} synonymously.

\begin{figure}[h]
	\centering
	\includegraphics[width=.45\columnwidth]{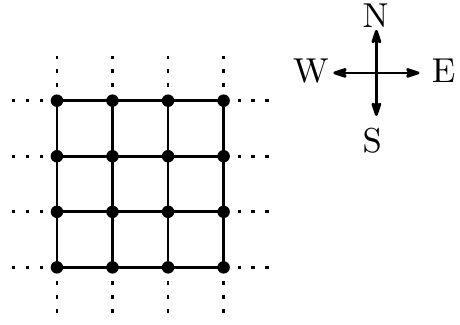}
	\caption[Definition Grid Graph]{An infinite square grid graph.}
	\label{def_gfx_grid}
\end{figure}

Every vertex of \(G\) is either \textit{occupied} by a tile or \textit{unoccupied}.
\textit{Tiles} represent passive particles of programmable matter and cannot move or manipulate themselves.
\Cref{def_gfx_tiles/P} (a) shows a graph where exactly one vertex is occupied.
A connected subset of vertices, where every vertex is occupied by a tile, forms a \textit{polyomino}, e.g., as shown in \Cref{def_gfx_tiles/P} (b).

\begin{figure}[H]\centering
	\subfigure[]{
		\includegraphics[scale=0.8]{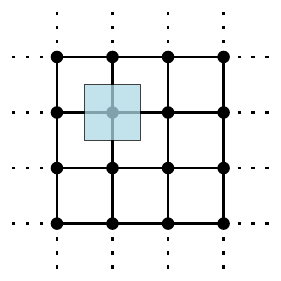}
	}\hfil
	\subfigure[]{
		\includegraphics[scale=0.8]{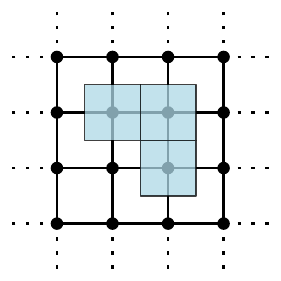}
	}
	\caption[Definition Tiles]{(a) One vertex occupied by a tile.
(b) A connected subset of vertices, in which every vertex is occupied by a tile, forming a polyomino.}
	\label{def_gfx_tiles/P}
\end{figure}

The \textit{boundary} of a polyomino \(P\) is denoted by \(\partial P \) and includes all tiles of \(P\) that are adjacent to an empty vertex (see also \Cref{def_gfx_partial-holes-comp} (a)).
Polyominoes can have \textit{holes}, i.e., there may be a single vertex, or connected sets of vertices that are empty and surrounded by tiles of \(P\).
\Cref{def_gfx_partial-holes-comp} (b) shows a polyomino with a hole.
Polyominoes without holes are called \textit{simple}; otherwise, they are \textit{non-simple}.
A polyomino consists of exactly one \textit{connected component} of tiles on the grid graph.
Hence, multiple separate connected components on the grid induce separate polyominoes.
\Cref{def_gfx_partial-holes-comp} (c) shows two connected components.

\begin{figure}[h]\centering
	\subfigure[]{
		\includegraphics[scale=0.73]{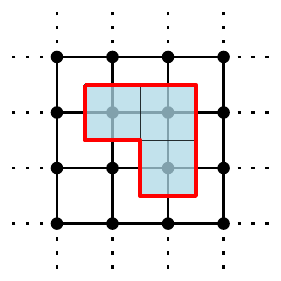}
	}\hfil
	\subfigure[]{
		\includegraphics[scale=0.73]{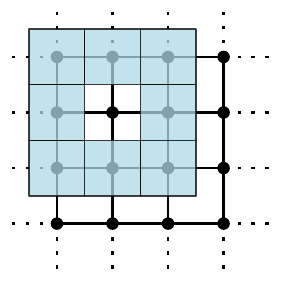}
	}\hfil
	\subfigure[]{
		\includegraphics[scale=0.73]{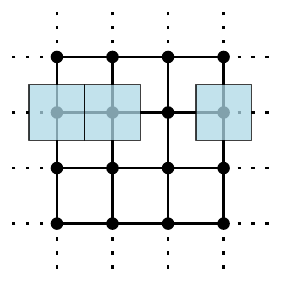}
	}
	\caption[Definition Polyominoes]{(a) The red line indicates the boundary of \(P\), denoted by \(\partial P\) (b) A non-simple polyomino with one hole within it.
(c) The tiles on the grid induce two separate connected components.}
	\label{def_gfx_partial-holes-comp}
\end{figure}

\textit{Monotony} describes a property of polyominoes.
In our case, we are interested in monotony with respect to the x- or y-axis.
A polyomino fulfills the monotony property for one axis if every line orthogonal to that axis intersects $\partial P$ at most twice.
\Cref{def_gfx_monotone} (a) shows an example for a non-monotone polyomino, (b) the case of an x-monotone polyomino and finally (c) an example for an x-y-monotone polyomino, also known as convex; the red-dotted-lines indicate orthogonal lines to one axis.

\begin{figure}[h]\centering
	\subfigure[]{
		\includegraphics[scale=0.6]{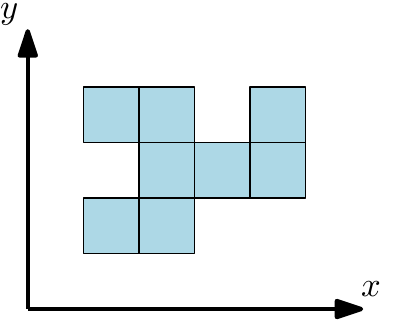}
	}
	\subfigure[]{
		\includegraphics[scale=0.6]{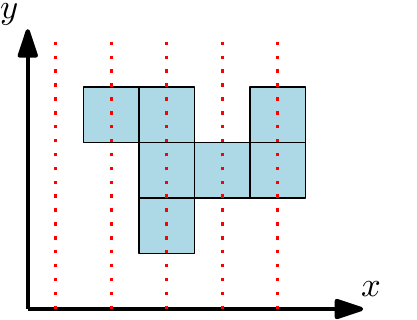}
	}
	\subfigure[]{
		\includegraphics[scale=0.6]{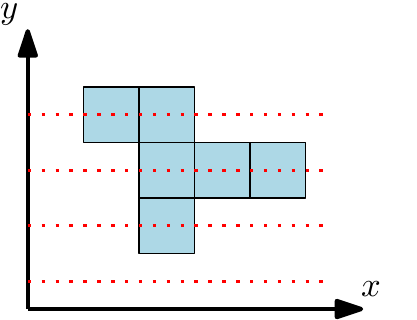}
	}
	\caption[Definition Monotony]{(a) A non-monotone polyomino.
(b) An x-monotone but not y-monotone polyomino.
(c) A convex polyomino.}
	\label{def_gfx_monotone}
\end{figure}

We use \textit{robots} as active particles in our model.
These robots work like \textit{finite deterministic automata} that can move around on the grid and manipulate the polyomino.
We indicate robots by a colored circle, as shown in \Cref{def_gfx_robot} (a).
A robot has the abilities to move along the edges of the grid graph and to change the state of the current vertex by placing or removing a tile on it.
Robots work by a series of Look-Compute-Move (LCM) operations:
Depending on the current state of the robot and the vertex it is positioned on (Look), the next step is computed concerning a specific transition function $\delta$ (Compute), which determines the future state of robot and vertex and the actual movement (Move).
Due to their finiteness, robots cannot perform operations that require non-constant memory, like counting or 
saving the coordinates of the current position.
In the case of multiple robots (\Cref{def_gfx_robot} (b)) we assume that they cannot be located on the same vertex at the same time.
Communication between robots is limited to adjacent vertices and can be implemented by expanding the look-phase by the states of all adjacent robots.
For the purposes of this paper, connectivity is ensured if the grid graph induced
by the union of all placed tiles and all robots is connected.
This implies that a robot can hold two components together, as shown in Figure~\Cref{def_gfx_robot} (c).

\begin{figure}[h]\centering
	\subfigure[]{
		\includegraphics[scale=0.7]{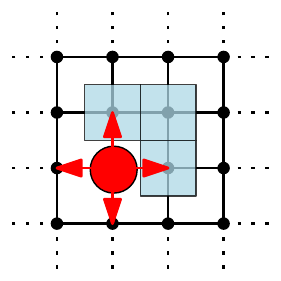}
	}\hfil
	\subfigure[]{
		\includegraphics[scale=0.7]{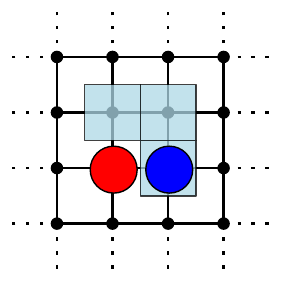}
	}\hfil
	\subfigure[]{
		\includegraphics[scale=0.7]{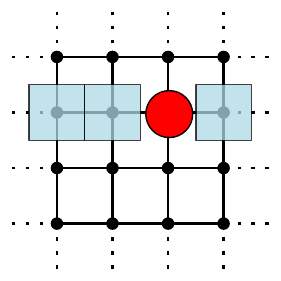}
	}
	\caption[Definition Robots]{(a) One robot and its moving possibilities.
(b) Two robots on the grid.
(c) Robots can hold separate connected components together.}
	\label{def_gfx_robot}
\end{figure}

\section{Constructing a Bounding Box}\label{BB}
Fekete et al.~\cite{Fekete2018} showed how to construct a bounding box around a polyomino.
However, this algorithm does not guarantee connectivity.
We describe an algorithm to construct a bounding box of width one, as shown in \Cref{BB_gfx_def}.
The bounding box of a given polyomino \(P\) will be denoted by \(bb(P)\) and consists of four edges or sides, one for every direction.
To accomplish the required connectivity between every placed tile, i.e., the bounding box, and \(P\), we specify without any loss of generality that the connection between \(bb(P)\) and \(P\) must be on the south side of the boundary.
For ease of presentation, we colorize the polyomino in blue and the bounding box in gray. 
The robots cannot distinguish between those tiles.

\begin{figure}[h]
	\centering
	\includegraphics[scale=0.7,page=13]{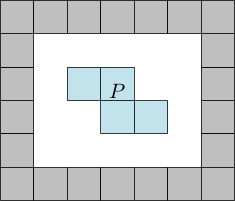}
	\caption[Definition Bounding Box]{A polyomino \(P\) and its bounding box, colored in gray.}
	\label{BB_gfx_def}
\end{figure}

We start with somef initial assumptions.

\begin{itemize}
	\item We assume that two robots are placed adjacent to each other on an arbitrary tile of the polyomino \(P\).
	\item The first robot \(R1\) is the leader, whose tasks are to find a possible starting position, construct the boundary and to decide whether a reached tile belongs to \(P\) or the already assembled bounding box.
	\item The second robot \(R2\) holds the polyomino and the bounding box together.
	The first boundary tile is placed right beneath this robot.
	Meeting the second robot from above corresponds to moving on the polyomino and reaching it from below indicates running on the bounding box.
\end{itemize}

\subsection{Construction}\label{BB_Construction}

The construction can be split into three phases.
The first is to locate a suitable starting position, the second covers the actual assembly process, and the last is to finish the construction.

\begin{itemize}
	\item FindStartingPosition:
	\begin{enumerate}
		\item \label{BB_ConstructionYMin} The first robot \(R1\) leads the search for an appropriate starting position.
That can be done by searching for a local y-minimal vertex, which is occupied by a tile.
The robots are always able to find a local y-minimum by repeating the following two steps until completion.
		
		\begin{enumerate}
			\item \(R1\) moves down until the last occupied vertex is reached.
			\item \(R1\) moves sideways to check any other column within that row for a possibly lower placed tile.
First in one direction and then in the other.
If there is no more column with a lower occupied vertex, the local y-minimum is found.
		\end{enumerate}
	
		The other robot \(R2\) follows \(R1\) by stepping into the direction, where \(R1\) was located.
When the robots reached the end of a column or row, it can be necessary to move a step back, i.e., to make a step in 
the opposite direction as before.
Therefore we need an additional communication step between every two move steps to ensure that both robots perform the same move.
		\item \(R2\) positions itself on the first vertex beneath this local y-minimal vertex.
Afterwards, \(R1\) starts the bounding box construction one vertex further down.
See \Cref{BB_gfx_StartPos} for an example initial arrangement.
	\end{enumerate}

	\begin{figure}[h]
		\centering
		\includegraphics[scale=0.7]{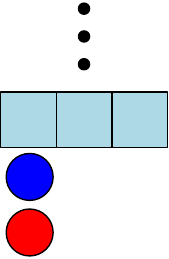}
		\caption[Starting Arrangement]{Arrangement of the two robots after a possible starting position is found.}
		\label{BB_gfx_StartPos}
	\end{figure}
	
	\item AssemblingBB:
	\begin{enumerate}
		\item The construction is performed clockwise around \(P\), i.e., whenever possible, \(R1\) makes a right turn.
		
		\item Once \(R1\) hits a tile during the construction, we want to know whether this is a tile of \(P\) or a tile of \(bb(P)\).
We describe this sub-procedure in \Cref{BB_Membership}.
If it is a tile belonging to \(P\), we need to shift the current line outwards until there is no more conflict and continue the construction.
\Cref{BB_gfx_shifting} shows an example; once we hit a tile, which belongs to \(P\) (\Cref{BB_gfx_shifting} (a)), we shift outwards until the assembly can be continued (b).
In case the hit tile belongs to the already assembled bounding box, we continue by finishing the construction process.
		
		\begin{figure}[h]\centering
			\subfigure[]{
				\includegraphics[scale=0.7]{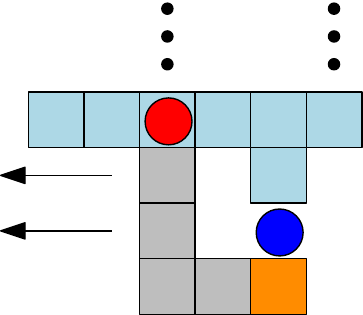}
			}\hfil
			\subfigure[]{
				\includegraphics[scale=0.7]{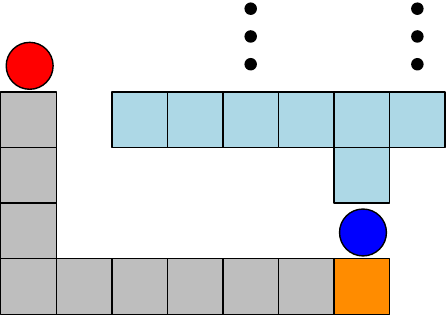}
			}
			\caption[Shifting Boundary Edge]{(a) \(R1\) hits a tile belonging to \(P\).
(b) The triggered shifting process is finished.}
			\label{BB_gfx_shifting}
		\end{figure}
	
		\item\label{BB_ConstructionBridge} In case the line to be shifted is the first edge of the constructed bounding box, we know that there exists a tile of \(P\) that has a lower y-coordinate than the current starting position.
So we build up a bridge to traverse this gap.
In the following, we describe the procedure of bridge building analogous to \Cref{BB_gfx_bridge}.
	
		\begin{enumerate}
			\item \(R1\) recognizes a tile belonging to \(P\).
			\item 
			\(R1\) removes the previously placed tile to mark the end of the constructed bridge.
Subsequently, it moves backward on the bridge to collect \(R2\).
At this moment, the bridge is only connected through the second robot on the previous starting position.
			\item \(R2\) places a tile on the vertex it stays on.
Both robots return to the end of the bridge and complete it.
Now we are connected at be beginning and the end of the bridge.
			\item \(R2\) returns to the previous starting position to destroy the bridge from its start.
Hence the robots continue by restarting the search for a starting position.
		\end{enumerate}
	
		\begin{figure}[h]\centering
			\subfigure[]{
				\includegraphics[scale=0.55]{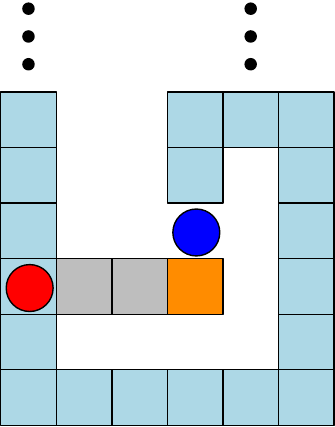}
			}\hfill
			\subfigure[]{
				\includegraphics[scale=0.55]{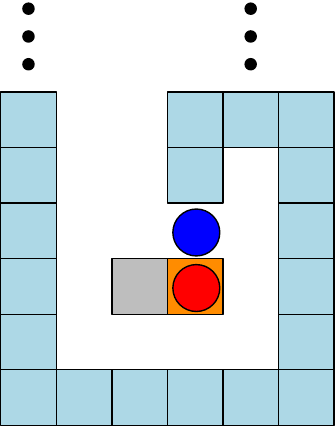}
			}\hfill
			\subfigure[]{
				\includegraphics[scale=0.55]{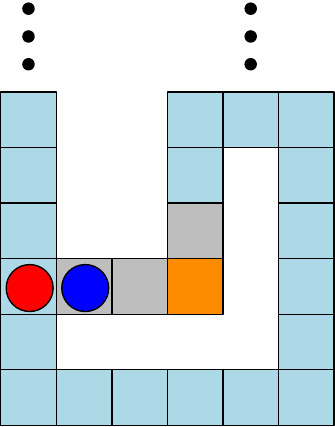}
			}\hfill
			\subfigure[]{
				\includegraphics[scale=0.55]{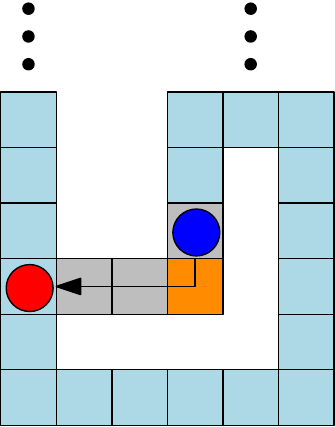}
			}
			\caption[Building Bridges]{Traversing a gap by building a bridge}
			\label{BB_gfx_bridge}
		\end{figure}
	\end{enumerate}

	\item Finishing:
	At some point, \(R1\) hits a tile that is part of the already assembled bounding box.
At that point, we can finish the bounding box construction.
	Two possible cases can occur:
	\begin{enumerate}
		\item \(R1\) hits the very first boundary tile, which is the tile beneath \(R2\), while building from east to west.
As it is shown in \Cref{BB_gfx_perfectFinish}, the bounding box is finished.
		\begin{figure}[h]
			\centering
			\includegraphics[scale=0.7]{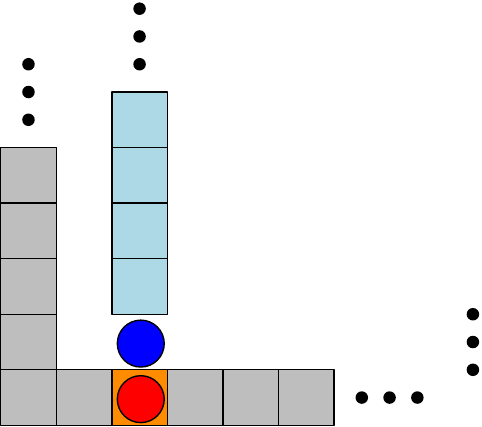}
			\caption{Perfect finish of bounding box construction}
			\label{BB_gfx_perfectFinish}
		\end{figure}
		
		\item\label{BB_FinishConstrCase2} \(R1\) hits the bounding box on any other tile, or the very first boundary tile while building from south to north, see \Cref{BB_gfx_finishing} (a) for an example.
If the hit tile is not a corner tile, the current line needs to be shifted outwards until the next corner is reached.
From that, we know that firstly exactly one three-way crossing exists and secondly there must be some part of the already assembled boundary that lies within the final bounding box that encloses \(P\).
We need to establish a new connection between \(bb(P)\) and \(P\) and remove the inner lying boundary tiles.
This can be done by executing the following steps, which are pictured in \Cref{BB_gfx_finishing}.
		\begin{itemize}
			\item \(R1\) moves clockwise around on the bounding box until the first tile of \(P\) with distance two to the southern side is found.
Starting with a configuration as in \Cref{BB_gfx_finishing} (b), the red robot would find the first suitable tile right above the position it is placed on in \Cref{BB_gfx_finishing} (c).
			\item \(R1\) ensures the connectivity between this tile and the boundary by placing one tile between both.
This step creates a second three-way crossing.
			\item \(R1\) retraces the path to \(R2\) by moving counterclockwise on the boundary and following every possible left-hand turn.
In \Cref{BB_gfx_finishing} (c) the path is indicated by the arrow.
			\item Once \(R2\) is fetched up, the part outgoing from the initially starting position can be destroyed until the first three-way junction is reached.
			\item Both robots move clockwise around on the bounding box searching for the only remaining three-way crossing, which was created a few steps before to ensure connectivity.
		\end{itemize}
		
		\begin{figure*}[h]
			\subfigure[]{
				\includegraphics[scale=0.5]{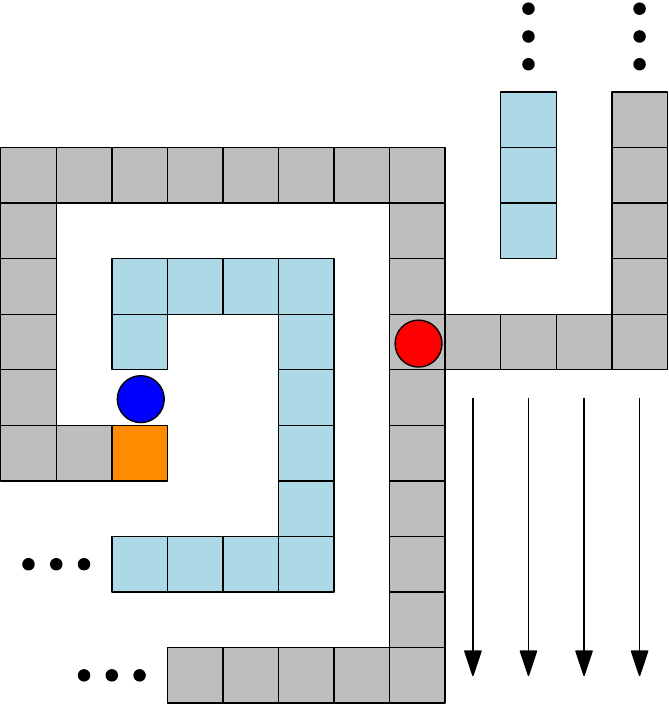}
			}\hfill
			\subfigure[]{
				\includegraphics[scale=0.5]{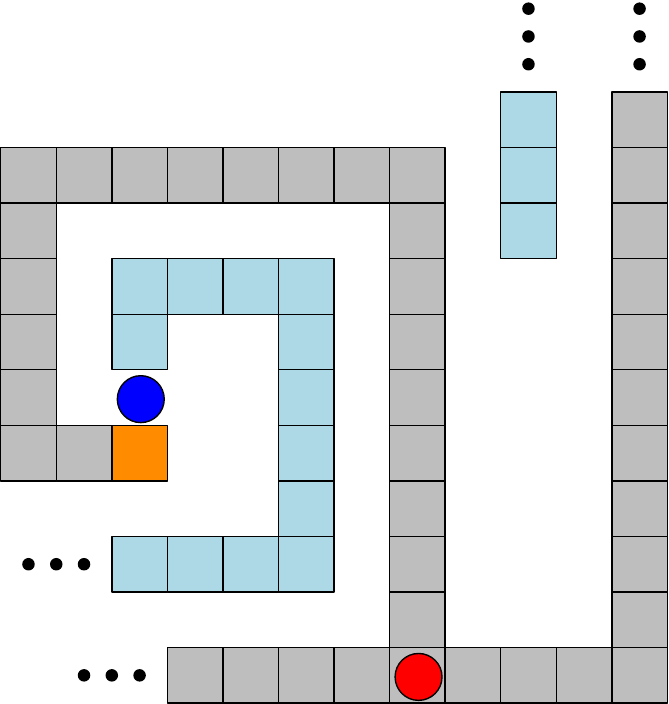}
			}\hfill
			\subfigure[]{
				\includegraphics[scale=0.5]{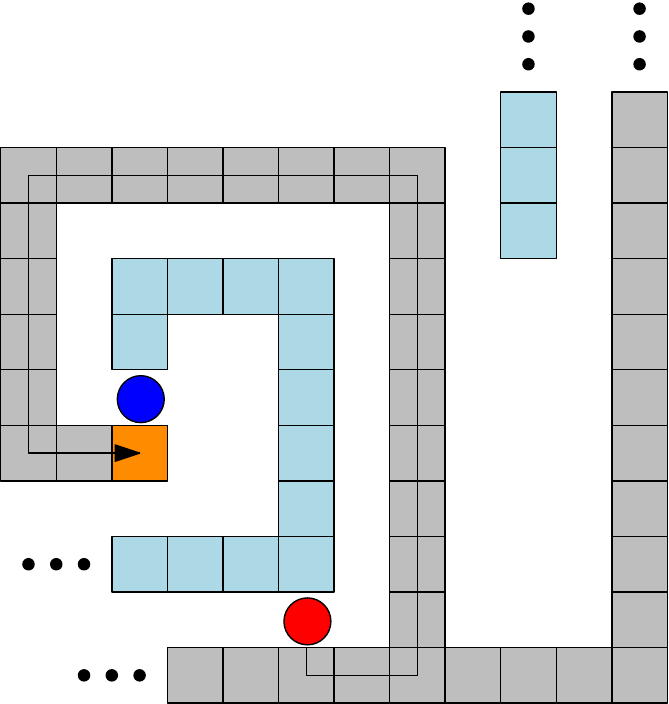}
			}\hfill
			\subfigure[]{
				\includegraphics[scale=0.5]{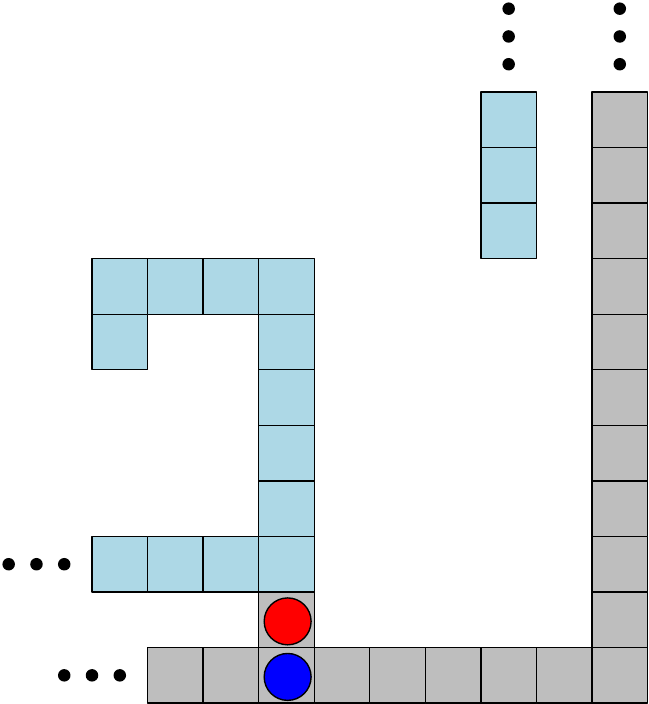}
			}
			\caption[Bounding Box Finish]{The second case of finishing the bounding box.
(a) An already constructed part of the bounding box is hit.
(b) The last boundary side is shifted.
(c) \(R1\) found a suitable new connectivity vertex above the southern side, places a tile and retraces path to the initial starting position.
(d) Unnecessary part of the bounding box is removed and both robots catch up to the new connection}
			\label{BB_gfx_finishing}
		\end{figure*}	
	\end{enumerate}

\end{itemize}

\subsection{Tile Membership}\label{BB_Membership}
If we hit any tile during the bounding box construction, we must know whether this tile belongs to the polyomino or the already assembled boundary.
To accomplish this, we make use of the following strategy.
Any possible bounding box tile has one blank adjacent vertex.
So if this is not given, we know without any further checks that the considered tile belongs to the polyomino.
Otherwise, we check the tile's membership as follows.
Whenever we hit a tile for which we would like to check membership (as for the tiles, the red robot \(R1\) is placed on in \Cref{BB_gfx_Membership-P} (a) and \Cref{BB_gfx_Membership-BB} (a)), the robot removes the previously placed tile (grey striped within both figures).
And, in case that was the first tile inducing a new side of the bounding box, i.e., it is the first tile after a right turn, we position a marker tile on the outer side of the boundary corner forming a left turn, see \Cref{BB_gfx_Membership-P} (b).
Since we only make right turns during the boundary construction, this will be the only discoverable left turn.
Hence, when we follow the bounding box beginning from the start, we can identify the position where the last bounding box tile was placed, either the vertex opposed to the only left turn or the first empty vertex when following the direction of the last side.
\Cref{BB_gfx_Membership-P} (b) shows the first and \Cref{BB_gfx_Membership-BB} (b) the second case.

\begin{figure}[h]\centering
	\subfigure[]{
		\includegraphics[scale=0.7]{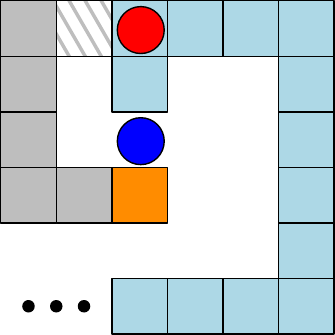}
	}\hfil
	\subfigure[]{
		\includegraphics[scale=0.7]{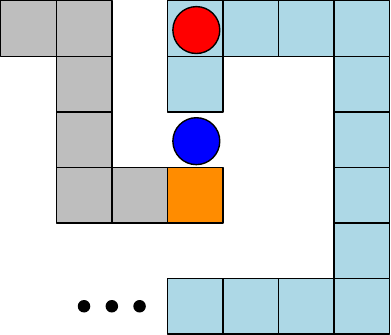}
	}
	\caption[Tile Membership \(P\)]{Checking membership of an arbitrary tile that belongs to the polyomino.
(a) The red robot is placed on the tile we have to check.
(b) Setting up the marker is done.
We can rediscover the progress of the current boundary construction process by searching for this left turn marker.}
	\label{BB_gfx_Membership-P}
\end{figure}

\begin{figure}[h]\centering
	\subfigure[]{
		\includegraphics[scale=0.7]{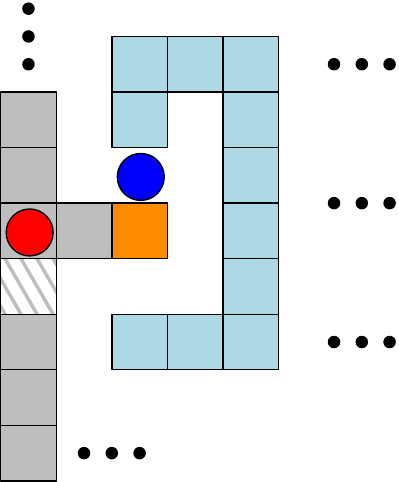}
	}\hfil
	\subfigure[]{
		\includegraphics[scale=0.7]{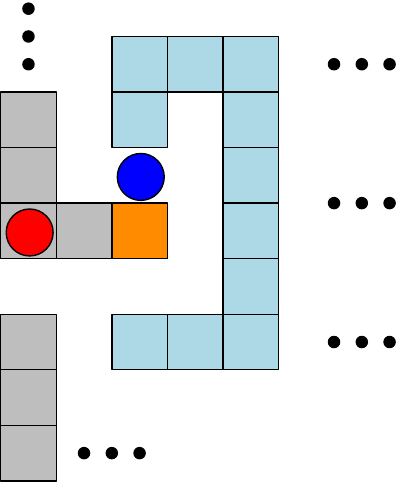}
	}
	\caption[Tile Membership \(bb(P)\)]{Checking membership of an arbitrary tile that belongs to the bounding box.
(a) The red robot is placed on the tile we have to check.
(b) In this case we only remove the previously placed tile.
The robot can rediscover the progress of the current bounding box construction process by following the direction of the last side and recognizing no left turn marker.}
	\label{BB_gfx_Membership-BB}
\end{figure}

From that the robot is going to move clockwise around  the boundary of the connected component, to which the tile we have to check belongs.
If the second robot \(R2\) is met from above, the tile to be tested belongs to the polyomino.
In the other case, we met \(R2\) from below, that indicates we tested a tile of the already constructed bounding box.
Both cases are shown in \Cref{BB_gfx_Membership-check} (a).

\begin{figure}[h]
	Case 1:
	\begin{center}
	\subfigure[]{
		\includegraphics[scale=0.7]{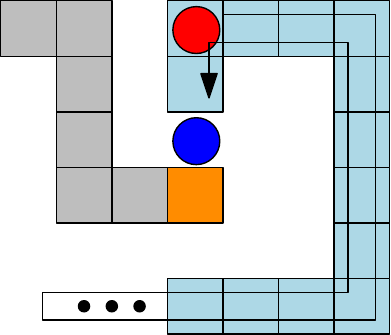}
	}\hfil
	\subfigure[]{
		\includegraphics[page=7,scale=0.7]{ipe/02-2_boundary-HitAnyTile-new.pdf}
	}
	\end{center}

	Case 2:
\begin{center}
	\subfigure[]{
		\includegraphics[scale=0.7]{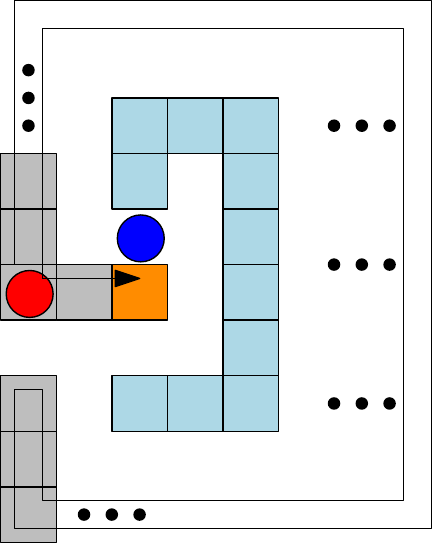}
	}\hfil
	\subfigure[]{
		\includegraphics[scale=0.7]{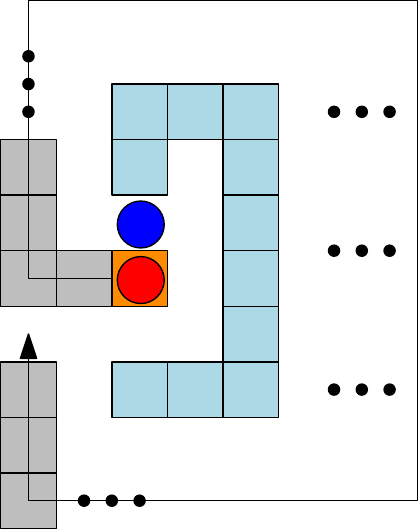}
	}
\end{center}
	\caption[Tile Membership Check]{(a) The arrows indicate the testing paths.
(b) Retracing to the checked tile with additionally knowledge about its membership.}
	\label{BB_gfx_Membership-check}
\end{figure}

In case the tested tile belongs to the polyomino, \(R1\) meets \(R2\) from above, which means \(R1\) moves on \(P\).
However, after we determined the tile's membership, we have to continue moving on the bounding box.
Hence we either switch the tasks of both robots and move synchronously one step downwards or \(R1\) steps along \(R2\) to reach the bounding box.
In case of a tested tile, which belongs to \(bb(P)\), \(R1\) is already on the boundary.
Afterward, we retrace to the position of the previously tested tile and continue the construction.

\subsection{Simple Polyominoes}\label{BB_Simple}

In case of a simple polyomino, the bounding box can be assembled using only one robot instead of two.
That is possible because we determine a bounding box structure, that enables the robot to distinguish between tiles belonging to the already constructed bounding box and those of the given polyomino.
Therefore the following adaptions are needed:

\begin{figure}[h]
	\centering
	\subfigure[]{
		\includegraphics[page=11,scale=0.8]{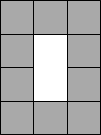}
	}\hfil
	\subfigure[]{
		\includegraphics[page=10,scale=0.8]{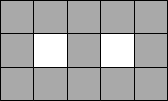}
	}\hfil
	\subfigure[]{
		\includegraphics[page=13,scale=0.8]{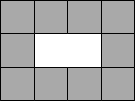}
	}	\hfil	
	\subfigure[]{
		\includegraphics[page=12,scale=0.8]{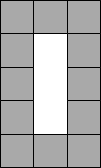}
	}
	\caption[Bounding Box Elements for Simple \(P\)]{Bounding box elements for simple polyominoes.
(a) Starting position.
(b) Bounding box elements.
(c) Bridge element.
(d) Connecting element for finishing bounding box construction.}
	\label{BB_gfx_simpleElements}
\end{figure}

\begin{itemize}
	\item Instead of creating a boundary of width one as described before, we construct it as a three tile wide line, consisting of an outer, middle and inner lane.
On the middle lane, every second vertex is unoccupied.
Since there are no holes within simple polyominoes, the robot can clearly distinguish between \(P\) and the bounding box.
(\Cref{BB_gfx_simpleElements} (b))

	\item The first bounding box element, which is placed right beneath a possible starting position has two vertically successive unoccupied vertices vertices.
(\Cref{BB_gfx_simpleElements} (a))

	\item To handle step \ref{BB_ConstructionBridge} of the assembling step in \Cref{BB_Construction}, that is traversing a gap with a bridge, it might be necessary to make use of a bounding box element with two horizontally successive unoccupied vertices.
The element is shown in (c) of \Cref{BB_gfx_simpleElements}, and an example for the purpose of application can be seen in \Cref{BB_gfx_bridgeSimple}.

	\begin{figure}[h]\centering
		\subfigure[]{
			\includegraphics[scale=0.6]{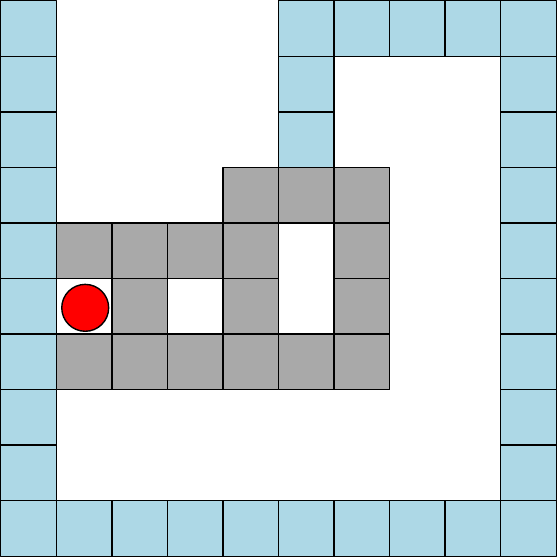}
		}\hfil
	\subfigure[]{
		\includegraphics[scale=0.6]{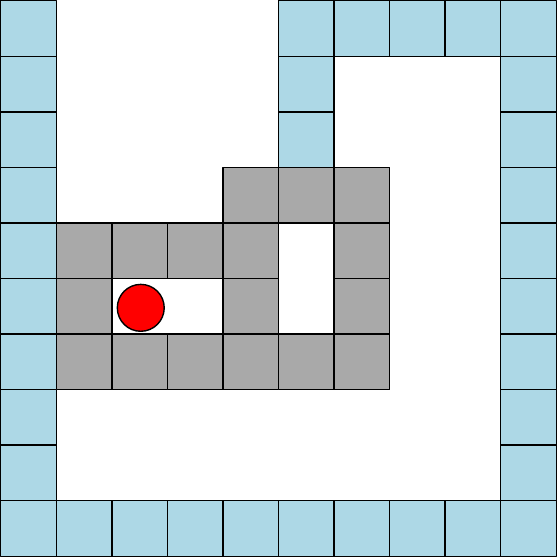}
	}
	\caption[Building Bridges Simple \(P\)]{Bridge construction in case of simple polyominoes.
(a) The bridge cannot be completed.
(b) Used the additional bridge element.}
		\label{BB_gfx_bridgeSimple}
	\end{figure}
	
	\item Because of the defined type of boundary elements, the minimum distance between polyomino and the constructed bounding box is either two or three and may differ for every side of the bounding box.
Since it might be necessary to search for a new connection between boundary and \(P\) while finishing the construction, and the minimal distance between the southern bounding box and \(P\) could be three, we need a second connection element, shown in Figure~\ref{BB_gfx_simpleElements} (d), with three vertically successive unoccupied vertices.
The use case can be seen in Figure~\ref{BB_gfx_finishingSimple}.

	\begin{figure*}[h]\centering
	\subfigure[]{
		\includegraphics[scale=0.55]{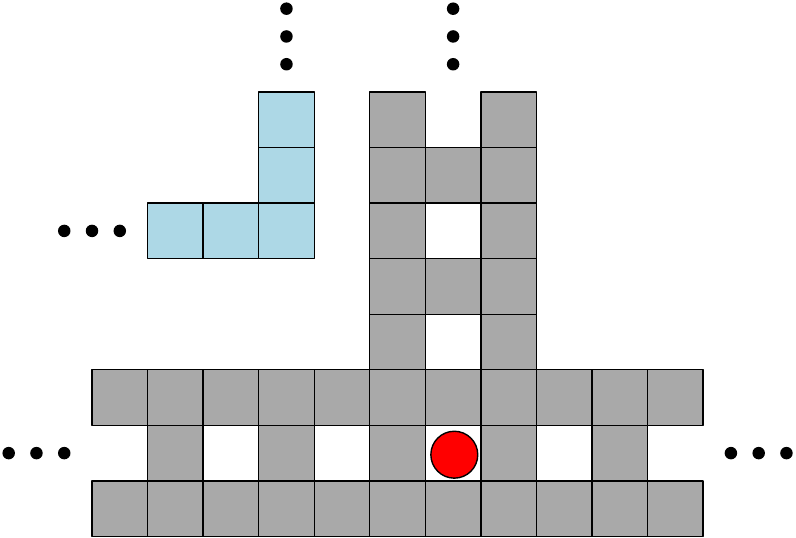}
	}\hfil
	\subfigure[]{
		\includegraphics[scale=0.55]{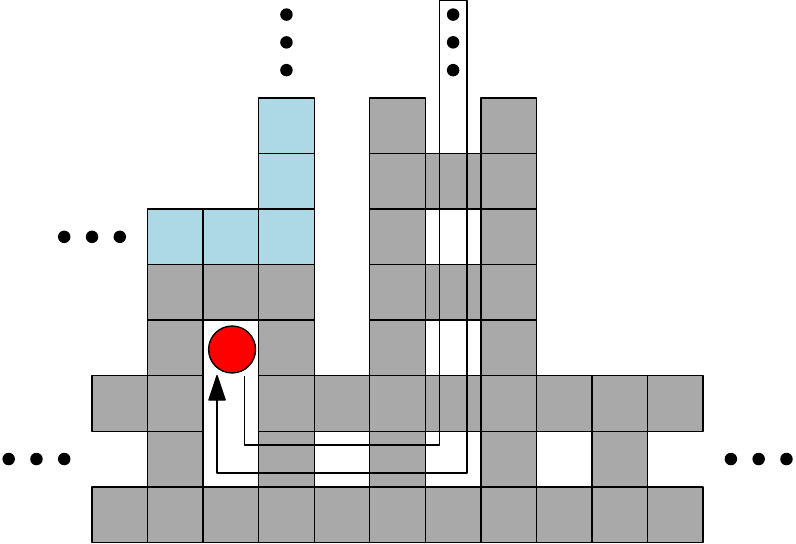}
	}\hfil
	\subfigure[]{
		\includegraphics[scale=0.55]{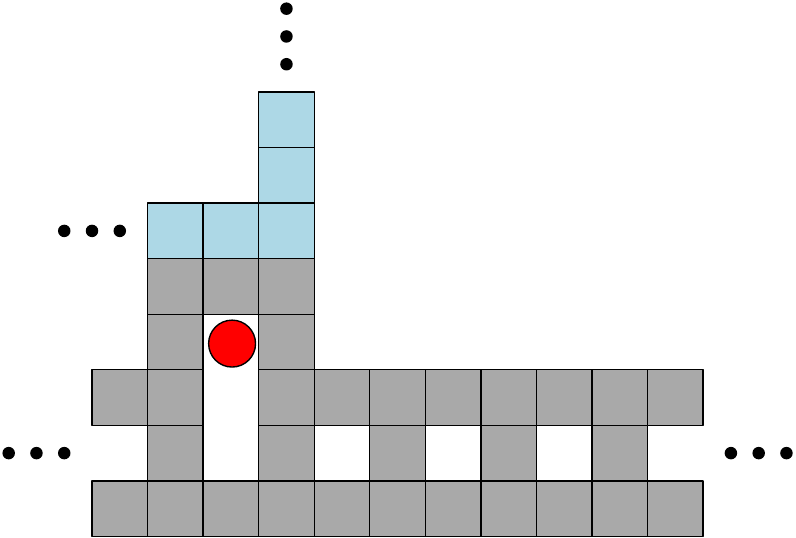}
	}
	\caption[Bounding Box Finish Simple \(P\)]{Finishing process in case of simple polyominoes.
(a) Search for a new connection.
(b) Used the additional connection element.
(c) Finished construction after retracing to initial start and removing unnecessary boundary elements.}
	\label{BB_gfx_finishingSimple}
	\end{figure*}
	
\end{itemize}

With this adaptions, we can build up a bounding box surrounding simple polyominoes with just one robot instead of two for the general case.

\subsection{Analysis}

\begin{theorem}\label{th:bounding_box}
	Given a polyomino \(P\) of width \(w\) and height \(h\), building a bounding box surrounding \(P\) with the need that boundary and \(P\) are always connected, can be done with two robots in
	\(O(max(w,h) \cdot (wh + k \cdot |\partial P|))\)
	steps, where \(k\) is the number of convex corners in \(P\).
\end{theorem}

\begin{proof}
	\textbf{Correctness:} 
	We consider an initial configuration.
The robots start anywhere on the polyomino, and we can always find a local y-minimum by moving down until the first empty is reached, followed by checking for possibly lower placed tiles within any other column of that row and in case it is, repeating these two steps.

	During the assembling step the robot, whenever possible, performs a right turn.
A right turn is possible if the minimal distance between the resulting boundary edge and \(P\) is two.
Every time the robot recognizes a tile of \(P\) with a smaller distance, the current line needs to be shifted outwards until there is no more conflicting tile.
If the line to shift is noticed as the first boundary side, the bounding box will be deleted entirely, and the construction process will restart after we searched for a new starting position.
After at most \(k\) turns the robot will either come into conflict, that results in restarting from another starting position, or it finds a part of the already assembled bounding box and continues by finishing the construction.
	Finishing the bounding box may require shifting the last assembled boundary edge to form a rectangle surrounding \(P\).
Thereby we built a three-way crossing which connects the outer bounding box and all boundary tiles lying inside of it as well as the second robot.
So we construct a second three-way crossing on the south boundary side via extending that side by placing one tile at some vertex, that lies between a tile of \(P\) and one of \(bb(P)\) and ensures connectivity between them.
Since there are no more than that two described three-way crossings, we can retrace the path from one to the other by moving counterclockwise (clockwise) taking every left-hand (right-hand) turn with attention on the characteristic, that the second junction extends the bounding box by just one tile to the inner.
So we can retrace to the starting position, pick up the second robot and remove all unnecessary boundary tiles.

	\textbf{Connectivity:} While assembling the bounding box or checking the membership of a tile, there is always one robot holding \(P\) and the bounding box together.
Whenever we have to build up a bridge for traversing a gap, the constructed bridge between the previous starting position and the first intersection with \(P\) will not be removed until one robot ensures connectivity at this intersection.
Finally, there is at least one connection between \(P\) and the bounding box when we are finishing the construction.
The assembly robot firstly places a tile as new connection right above the south boundary edge and only then retraces to the previous starting position to pick up the second robot and disassembles the inner bounding box parts.
	
	\textbf{Time:} Analogous to \cite{Fekete2018}, we show that every within the final bounding box is visited constant times, but there may be vertices that can be visited up to \(max(w,h)\) times.
Consider a vertex \(v\) within the final bounding box.
\(v\) will be visited for the first time while constructing the boundary.
Other visits are possible during shifting and removing the boundary tile placed on \(v\).
Finally, there may occur a constant number of more visits while building a bridge for searching a new starting position.
That results in \(O(wh)\) steps.
The second thing is that there may be a vertex \(v\) that is visited up to \(max(w,h)\) times.
\Cref{BB_gfx_worstCase} shows a worst-case example.
Since there are possibly up to \(max(w,h)\) starting positions \(s_i\) and the vertex \(v\) could be visited after every start or restart, \(v\) can be visited up to \(max(w,h)\) times.
This yields in \(O(max(w,h) \cdot wh)\) steps for the construction process.
For every tile that leads to a conflict during construction, the decision of whether that tile belongs to \(P\) or the bounding box is done by using the subroutine described in \Cref{BB_Membership}.
This subroutine takes \(O(|\partial P|)\) steps for tiles of \(P\) and \(O(|bb(P)|)\) steps for boundary tiles.
For every tile in \(\partial P\) we have a constant number of tiles in \(bb(P)\) and therefore \mbox{\(O(|bb(P)|) \subseteq O(|\partial P|)\)}.
We only check the membership of a hit tile if the conflict appears while the robot is in assembling state and not while shifting an already existing boundary part.
This can only occur up to \(k\) times because just after a convex corner is passed, the current side can be extended.
Therefore we have a total of \(O(max(w,h) \cdot (wh + k \cdot |\partial P|))\) steps.
\end{proof}

\begin{figure}[h]
	\centering
	\includegraphics[scale=0.6]{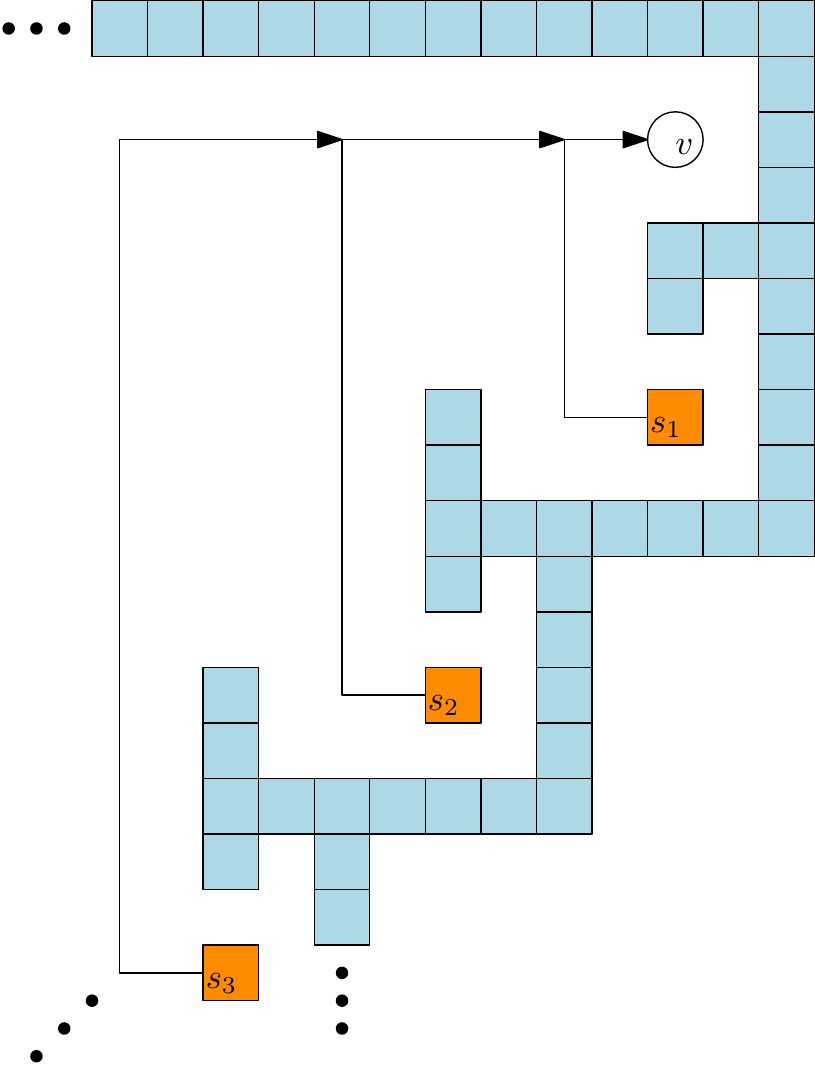}
	\caption[Bounding Box Worst-Case Example]{The vertex \(v\) is visited once for every starting position \(s_i\).}
	\label{BB_gfx_worstCase}
\end{figure}
In case of simple polyominoes and due to the possibility of clearly distinguish between tiles of the already assembled bounding box and those of the polyomino, the membership check, as described in \Cref{BB_Membership}, is not needed anymore.
The distinction can be done by just one robot and hence there is no need for additional robots.
So we obtain:

\begin{corollary}\label{cor:bounding_box_simple}
	Given a simple polyomino \(P\) of width \(w\) and height \(h\), building a bounding box surrounding \(P\) with the need that boundary and \(P\) are always connected, can be done with one robot in
	\(O(max(w,h) \cdot wh)\)
	steps.
\end{corollary}

\section{Scaling Polyominoes}
\label{Scale}

In this section, we present several algorithms.
The first is up-scaling of polyominoes in the general case when there already exists a bounding box surrounding the given polyomino.
This is followed by an algorithm for the inverse procedure, down-scaling a polyomino that is already scaled.
Finally, we present a strategy for adapting algorithms that solve arbitrary
tasks on a polyomino within the same model, with the additional constraint of preserving connectivity.

\subsection{Up-Scaling}

We describe a strategy for scaling polyominoes by a constant scaling factor \(c\).
Our strategy ensures that all placed tiles and all robots are connected during the whole process.
The scaled version will be build up right next to the given polyomino.
In the following, we describe the case of constructing the scaled version of \(P\) to the left of its initial position; the other case follows analogously.

\subsubsection{Preparation}\label{Scale_Prep}
\Cref{Scale_gfx_init} shows an example configuration before starting the scaling process.
Any initial configuration fulfills the following conditions.

\begin{itemize}
	\item The one tile wide bounding box that encloses the polyomino \(P\) is already constructed.
There is no edge (N, E, S, W) of the bounding box with a minimal distance to \(P\) that is greater than two.
	\item We start with the leader (red) robot \(R1\) right beneath the robot \(R2\) that holds \(P\) and the boundary together.
This connection between \(P\) and the bounding box is located on the south side of the bounding box, as defined in \Cref{BB}.\\
\end{itemize}

\begin{figure}[h]
	\centering
	\includegraphics[scale=0.6]{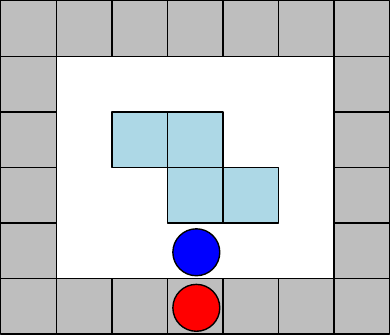}
	\caption[Initial Configuration]{Example configuration before starting the scaling process.}
	\label{Scale_gfx_init}
\end{figure}

During the scaling process, three columns within the bounding box (including the box itself) are involved.
The first (from west to east) is the current column of \(P\) to scale; the column where the robot is placed on in \Cref{Scale_gfx_prep}.
The second column, which is filled with tiles excepting the topmost row, is used to ensure connectivity and helps to recognize finishing of the current column.
And the third column marks both, the current overall progress and that within the column to scale.
This three columns can be prepared as follows:

\begin{enumerate}
	\item\label{Scale_Prep1} 
	The robot fills up the easternmost column within the bounding box, excepting the northernmost vertex.
From now on there is no other robot needed, cause this filled column ensures the connectivity of the bounding box and \(P\).
	\item\label{Scale_Prep2} The robot walks southwards and removes the southernmost tile of the column to the right.
This marks that the current column to scale lies two steps to the west.
For marking the next tile to scale, he removes the tile two steps to the north, which is the first possible tile of \(P\).
	\item\label{Scale_Prep3} Place a tile right above the south side of the bounding box in the column to scale.
\end{enumerate}

\begin{figure}[h]
	\centering
	\includegraphics[scale=0.6]{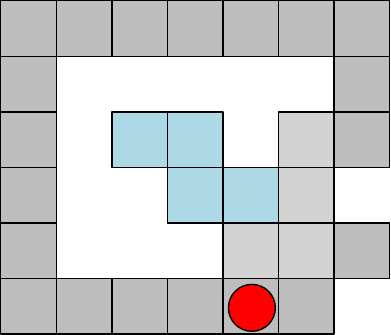}
	\caption{Prepared example configuration.}
	\label{Scale_gfx_prep}
\end{figure}

\subsubsection{Scaling Process}
\label{UpScaling}

It follows a description of the different states during the scaling process.
All example figures are based on a scaling factor \(c=3\).

\begin{itemize}
	\item \textit{NextTile}: The robot moves right until the column marker is reached, that is the first empty vertex.
And then the robot moves upwards looking for the row marker.
The next tile to handle is the one lying two steps to the left from this marker.
	Outgoing from the row marker, three cases may appear:
	\begin{enumerate}
		
		\item The vertices one and two steps to the left are occupied.
The robot shifts the row marker one vertex up and continues by placing a scaled \(c \times c\) segment, that represents the current occupied vertex, in the target area.
(ScaleTile)
		
		\begin{figure}[H]
			\centering
			\includegraphics[scale=0.6]{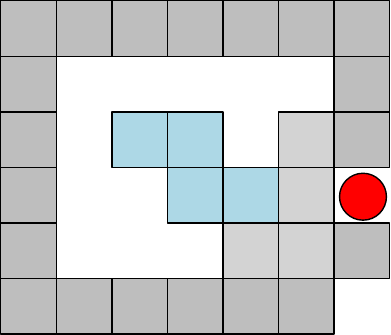}
			\caption{'ScaleTile' case}
			\label{Scale_gfx_scaleTile}
		\end{figure}
	
		\item The vertex one step to the left is occupied and the vertex two steps to the left is empty.
The robot places a tile on the inspected vertex, shifts the row marker one vertex up and continues with placing a scaled \(c\times c \) segment, that represents the current blank vertex, in the target area.
(ScaleEmpty)
		
		\begin{figure}[H]
			\centering
			\includegraphics[scale=0.6]{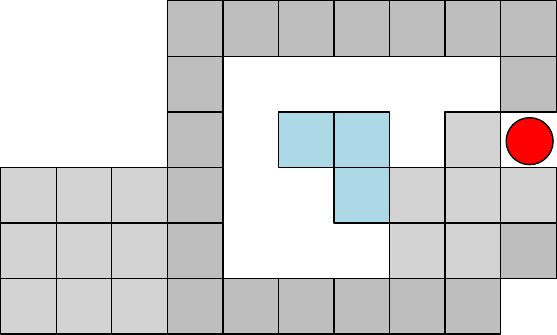}
			\caption{'ScaleEmpty' case}
			\label{Scale_gfx_scaleEmpty}
		\end{figure}
	
		\item The vertex one step to the left is empty.
That indicates the current column is finished.
(NewColumn)
		
		\begin{figure}[H]
			\centering
			\includegraphics[scale=0.6]{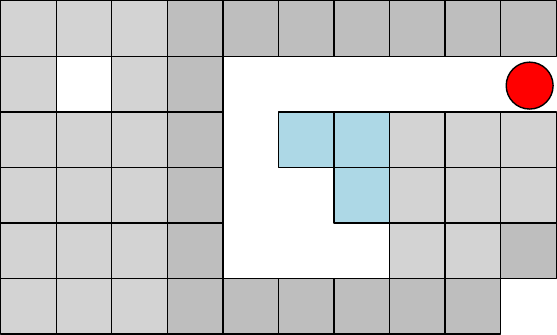}
			\caption{'NewColumn' case}
			\label{Scale_gfx_newColumn}
		\end{figure}
	
	\end{enumerate}

	\item \textit{ScaleEmpty/ScaleTile}: The robot builds up a scaled version of the inspected vertex to the left of the bounding box.
The exact procedure depends on whether it was the first vertex in a new column or any other vertex.
For the first case, the robot follows the row induced by the southern bounding box side leftwards until an unoccupied vertex appears and starts building the scaled element.
Otherwise, after detecting the first unoccupied vertex, the robot moves upwards on the last occupied vertex until the next blank vertex appears, that is the final position for placing the scaled element.\\
	The scaled version of a tile is represented by a $c\times c$ square of occupied vertices.
To label an scaled empty vertex, one vertex or any other sized square of vertices in the center of the $c\times c$
scaled element has to stay unoccupied.
With the exception of scaling factor \(c=2\), where one chosen vertex of the $2 \times 2$ segment stays blank.\\
	After the scaled element is placed, the robots state transits back to 'NextTile'.
	
	\item \textit{NewColumn}: Preparing a new column can be done similarly to steps \ref{Scale_Prep2} and \ref{Scale_Prep3} of preparation in \Cref{Scale_Prep}, i.e., we remove the previously used marker, place them one column further to the left and position a tile on the lowermost vertex within the new column.
The scaling process is finished, when the desired vertex during adapted step \ref{Scale_Prep3} is already occupied as shown in \Cref{Scale_gfx_newColCleanUp}.
That will only appear if the left side of the bounding box is reached and lets the robots state transit to 'CleanUp' state.
Otherwise, the robot continues with the next tile to scale.
	
\begin{figure}[h]
	\centering
	\includegraphics[scale=0.6]{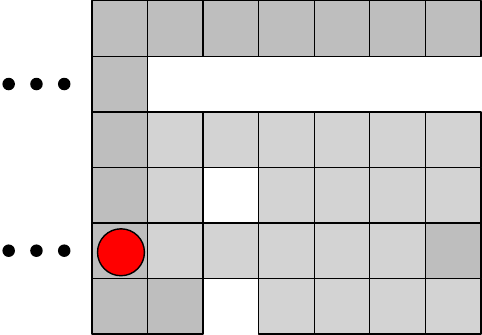}
	\caption['CleanUp' State]{Configuration that causes the robot transiting to 'CleanUp' state.}
	\label{Scale_gfx_newColCleanUp}
\end{figure}

	\item \textit{CleanUp}: The clean up subprocess can be split into two steps.
	\begin{enumerate}
		\item Removing all tiles within the originally \(w \cdot h\) bounding box area including the boundary leftovers themselves.
This can be done columnwise from right to left until the first column without a single blank vertex occurs.
As seen in \Cref{Scale_gfx_cleanUp1} the first column with this condition can be identified as the left boundary side.
		
		\begin{figure}[H]
			\centering
			\includegraphics[scale=0.6]{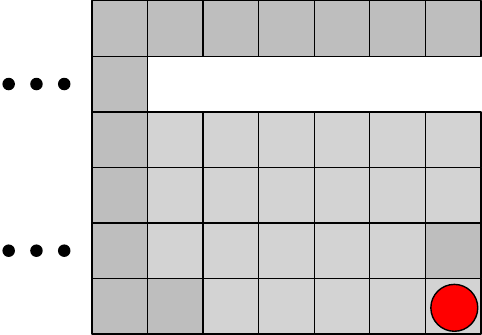}
			\caption['CleanUp' Step One]{First step during clean up.}
			\label{Scale_gfx_cleanUp1}
		\end{figure}
	
		\item Removing all \(c\times c \) segments that represent an unoccupied vertex within the target area.
\Cref{Scale_gfx_cleanUp2} shows the start of this step (a) and the finished clean up process (b).
Because the last scaled column was completely blank, this is done from left to right.
Every time we move onto a new column, the robot has to check if this is the last column.
That is needed because in case it is, we have to clean up from both sides of that column, i.e., top and bottom, until the first segment is reached, which represents an occupied vertex.
Afterward, we have to check for possibly further segments which may have to be removed between these two.
This special handling of the last column is needed, cause if we just follow a top-down or bottom-up strategy, we could lose connectivity.\\
		
	\end{enumerate}

	\begin{figure}[h]\centering
		\subfigure[]{
			\includegraphics[scale=0.6]{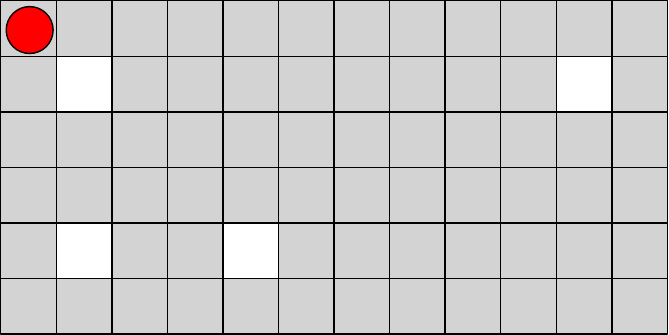}
		}\hfil
		\subfigure[]{
			\includegraphics[scale=0.6]{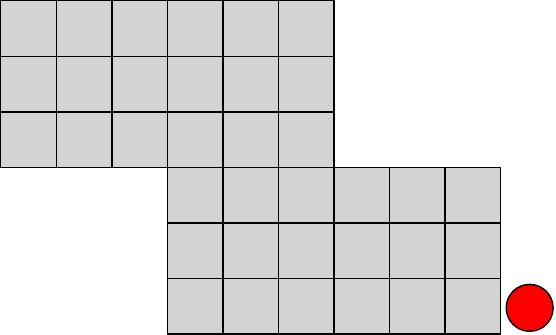}
		}
		\caption['CleanUp' Finished]{(a) Starting step two from left to right; (b) 'CleanUp' is finished and the completely scaled polyomino is left.}
		\label{Scale_gfx_cleanUp2}
	\end{figure}
\end{itemize}

\subsubsection{Simple Polyominoes}
In case the polyomino to be scaled is simple and the bounding box was constructed as described in section \ref{BB_Simple}, minor adaptions are needed.

Because there could be a minimal distance of three between bounding box and polyomino, the robot has to check for connectivity with \(P\) while performing step \ref{Scale_Prep1} of \Cref{Scale_Prep}, which means that the robot has to fill up a second column if connectivity is not ensured.
After that it is possible to remove the connection segment between southern boundary and \(P\).
For the further scaling process it is needed to determine the inner or outer lane of the boundary to hold the column marker.

\subsubsection{Analysis}
\begin{theorem}\label{th:scaling}
	Scaling a polyomino \(P\) of width \(w\) and height \(h\) by a constant scaling factor \(c\) without loss of connectivity can be performed with two robots in \(O(wh \cdot (c^2 + cw + ch))\) steps.
\end{theorem}

\begin{proof}
	\textbf{Correctness:} Starting with a described initial configuration, connectivity between the polyomino and the constructed bounding box is guaranteed by one robot holding both together.
After the preparation steps are done, there is at least one tile in the current column connecting the polyomino and all other placed tiles at any time.
Because every scaled segment is placed right next to the bounding box or an already scaled element of \(P\), connectivity is ensured during the scaling process.
Cleaning up is done by removing every tile within the initial bounding box area as well as the remaining parts of the boundary itself.
Followed by removing every element within the scaled \(c\cdot w\times c\cdot h\) area that represents an empty vertex.
Both can be done as described from one side to the other and will always guarantee connectivity.
The scaling step is performed from right to left and terminates when the left side of the bounding box is detected.
Because of the need for a bounding box surrounding \(P\) and we are only able to construct it with two robots, we also need two robots for scaling polyominoes under the condition of connectivity.

	\textbf{Time:} Every of the \(w \cdot h\) vertices within the bounding box of \(P\) will be handled twice during the whole process.
First time for checking whether there is a tile placed on it or not and the second time during clean up after the scaled version is completed.
The needed steps of scaling one vertex are given by searching for the desired target position and placing a segment of \(c \times c\) size.
Searching the target position can be done in \mbox{\(O((w + h) + (c\cdot w + c\cdot h)) \subset O(c\cdot w + c\cdot h)\)} steps and constructing the desired segment needs \(c^2\) steps.
That results in \(O(c^2 + cw + ch)\) steps per vertex.
The clean up requires visiting every vertex within the initial bounding box area and the scaled polyomino once, which stops after \(w\cdot h \cdot c^2\) steps.
Hence we have a total of \(O(wh \cdot (c^2 + cw + ch))\) steps for scaling a polyomino.
\end{proof}

Due to the possibility of constructing the bounding box for a simple polyomino \(P\) with only one robot as described in \Cref{BB_Simple} and the fact that additional robots are not needed for the described scaling strategy, we obtain:
\begin{corollary}\label{cor:scaling_simple}
	Scaling a simple polyomino \(P\) of width \(w\) and height \(h\) by a constant scaling factor~\(c\) without loss of connectivity can be performed with one robot in \(O(wh \cdot (c^2 + cw + ch))\) steps.
\end{corollary}

\subsection{Down-Scaling}
\label{DownScaling}
Down-scaling a polyomino that is already scaled by a known \(c\) corresponds to scaling it by the factor \(\frac{1}{c}\) and works similarly to the previously described up-scaling.
We position the down-scaled version right beside the up-scaled one.
In the following, we describe the case of assembling the down-scaled version of \(P\) to the right, the other case follows analogues.
Prerequisite is that the \(cw \times ch\) area is filled with segments representing an occupied or unoccupied vertex.
If this is not given, we need to construct a bounding box and fill up every empty \(c \times c\) square within the bounded area with a segment that represents an unoccupied vertex.
Afterward, the bounding box and rows or columns which only contain segments representing empty vertices, have to be removed.
From that, we obtain the previously described rectangle shape which has to be handled.
We mark the first segment and prepare the target area by extending the lowermost row by three tiles, directed to the target area.
\Cref{DownScale_gfx_Start} shows (a) a polyomino satisfying the prerequisites and (b) the initial configuration.
The rightmost column of the target area is used as a marker column.
The next desired target vertex can be determined by searching for the first empty vertex in this marker column and outgoing from there it is the vertex one step to the left.

\begin{figure}[h]\centering
	\subfigure[]{
		\includegraphics[scale=0.6]{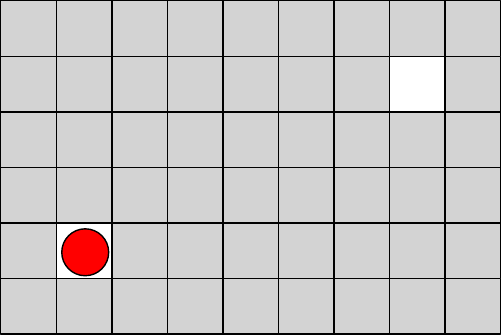}
	}\hfil
	\subfigure[]{
		\includegraphics[scale=0.6]{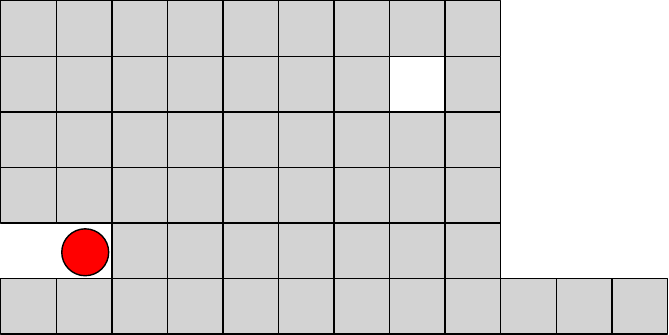}
	}	
	\caption[Starting Down Scaling]{(a) A polyomino that satisfies the described prerequisites.
(b) Starting arrangement for down scaling where the target area lies to the right of \(P\).}
	\label{DownScale_gfx_Start}
\end{figure}

From that we are going to process \(P\) columnwise from left to right and bottom-up within the columns.
To track the progress within a column we place a marker in the current \(c \times c\) segment.
For the first column we place tiles for segments of \(P\) representing an occupied vertex and for the further columns we remove tiles for segments of \(P\) representing an empty vertex, because the desired target column is completely filled with tiles as it was used as marker column before; the difference is shown in \Cref{DownScaling_gfx_diffColumns}.
This is needed to ensure connectivity during the execution.
After the down-scaling of a column is finished, the robot removes its initial version and extends the target area by one tile as a new marker column.
After the last column of the scaled version of \(P\) has been destroyed, the robot recognizes that there is no further valid \(c \times c\) segment in the next column, as shown in \Cref{DownScaling_gfx_finish}.
Hence we can finish the process by removing the farthest column and the lowest row.

\begin{figure}[h]\centering
	\subfigure[]{
		\includegraphics[scale=0.6]{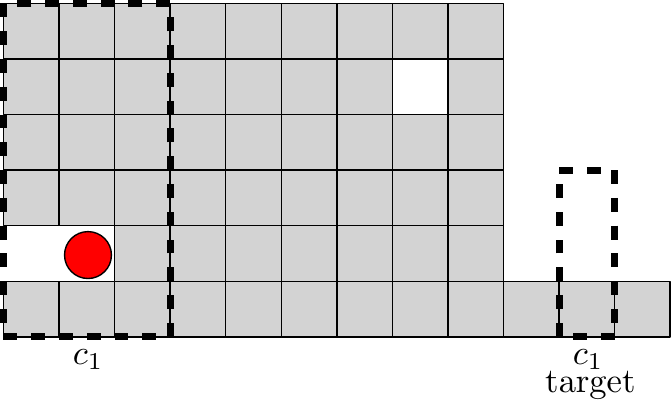}
	}\hfil
	\subfigure[]{
		\includegraphics[scale=0.6]{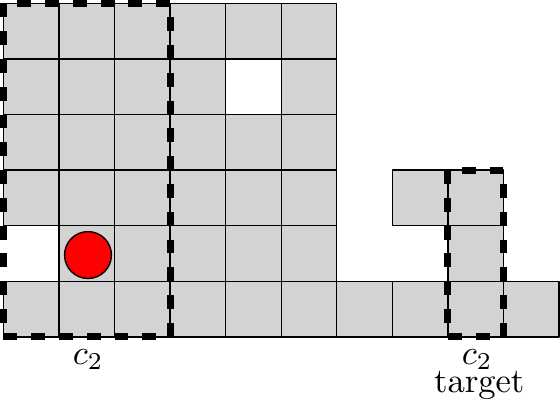}
	}
	\caption[Down Scaling First and Further Columns]{(a) \(c_1\) and its target.
We have to place tiles for every occupied segment.
(b) \(c_2\) and its target.
We have to remove tiles for unoccupied segments}
	\label{DownScaling_gfx_diffColumns}
\end{figure}

\begin{figure}[h]\centering
		\subfigure[]{
			\includegraphics[scale=0.8]{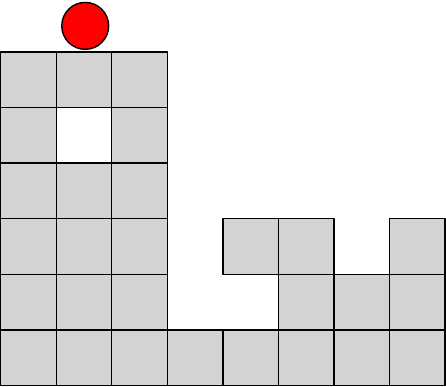}
		}\hfil
	\subfigure[]{
		\includegraphics[scale=0.8]{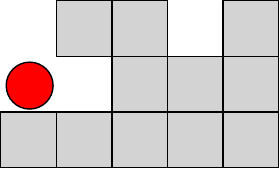}
	}
	\caption[Down Scaling Finish]{(a) The last column is finished.
(b) The robot recognizes no more valid segments in the next column.}
	\label{DownScaling_gfx_finish}
\end{figure}

\begin{corollary}\label{cor:down_scaling}
	Given a scaled polyomino \(P\) of width \(cw\), height \(ch\).
Down scaling \(P\) by the inverted constant scaling factor \(\frac{1}{c}\) without loss of connectivity can be performed with one robot in \mbox{\(O(wh \cdot (c^2 + cw + ch))\)} steps.
\end{corollary}

\subsection{Adapting Algorithms}
\label{AdaptAlg}

Within the Robot-on-Tiles model there are different algorithms solving various tasks on polyominoes.
These algorithms not necessarily ensure connectivity between all tiles and robots during execution.
We show that with the help of the described scaling strategy, we can adapt most of these algorithms 
to preserve connectivity.

\begin{definition}
	We define the \textit{robots movement area} (striped in \Cref{Def_gfx_space}) while executing an arbitrary algorithm as a \(w'\times h'\) rectangle in which \(w'\) is defined as the difference between the lowest and highest x-coordinate out of all visited vertices and all placed tiles; \(h'\) is defined analogously for y-coordinates.
\end{definition}

\begin{figure}[h]
	\centering
	\includegraphics{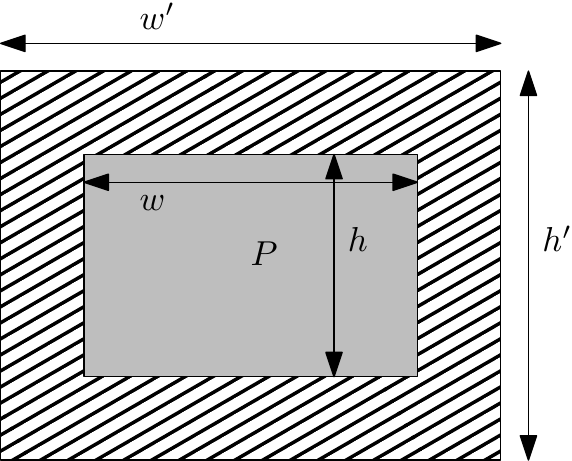}
	\caption[Space Requirement]{A polyomino \(P\) of width \(w\) and height \(h\) and the robots movement area while executing an arbitrary algorithm}
	\label{Def_gfx_space}
\end{figure}

\begin{theorem}\label{th:adapting}
	Given an algorithm \(\mathcal{A}\) for solving an arbitrary task on a polyomino \(P\) of width \(w\) and height \(h\) within the Robot-on-Tiles model, a time complexity \(\mathcal{T(A)}\) and which possibly not ensures connectivity between all placed tiles, we can generate an algorithm \(\mathcal{A'}\) for the same task with additionally guaranteed connectivity during the whole procedure and a total running time of \(O(wh\cdot (c^2 + cw + ch) + max((w'-w)h',(h'-h)w') + c\cdot \mathcal{T(A)})\)
\end{theorem}

\begin{proof}
	\textbf{Approach:} 
	The first step is to scale the polyomino \(P\) of width \(w\) and height \(h\) by any constant factor \(c > 1\).
We make use of the described strategy (\Cref{UpScaling}), with a minor adaption during the clean up.
We only remove the leftovers of the initial polyomino and the constructed bounding box, leaving the area where the scaled version was built up as it is.
Hence we get a \(cw \times ch\) rectangle containing the scaled polyomino and \(c \times c \) segments representing empty vertices from within the \(w \times h\) area of the initial polyomino.
	From that for every single robot movement in \(\mathcal{A}\), the robot makes \(c\) steps in \(\mathcal{A'}\) and if necessary additionally checks whether it is placed on a \(c \times c\) segment representing an empty or an occupied vertex.
\Cref{AdaptAlg_gfx_movement} (a) shows an example, where the robot moves along the red lines and stops on the black doted vertices within a \(c \times c\) segment.
	Every time the robot steps on an empty vertex, which does not belong to a scaled \(c \times c\) segment (\Cref{AdaptAlg_gfx_movement} (b)), we fill up the entered column or row with \(c \times c\) segments, representing unoccupied vertices.
To identify that segment of the column (row) which we stepped on firstly, we mark it by leaving an extra vertex within that segment blank, as shown in \Cref{AdaptAlg_gfx_fillRow}.\\
	\textbf{Analysis:} The scaling procedure requires \(O(wh\cdot (c^2 + cw + ch))\) steps.
Due to the described movement adjustment of doing \(c\) steps in \(\mathcal{A'}\) instead of one step in \(\mathcal{A}\), the running time of \(\mathcal{A}\) increases by the constant factor \(c\).
Finally there is the fill up procedure, which we use in case of stepping on an empty vertex.
Depending on the robots movement area for \(\mathcal{A}\), we get \(O(max((w'-w)h',(h'-h)w'))\) steps for filling up.
This results in an overall running time of \(O(wh\cdot (c^2 + cw + ch) + max((w'-w)h',(h'-h)w') + c\cdot \mathcal{T(A)})\) steps.
\end{proof}

After the algorithm terminates and if there exists only one polyomino in its scaled representation within the constructed rectangle area, we can eliminate empty columns and rows followed by using the described down-scaling strategy (\Cref{DownScaling}) for resetting it to its initial scale.

\begin{figure}[h]\centering
		\subfigure[]{
			\includegraphics[scale=0.7]{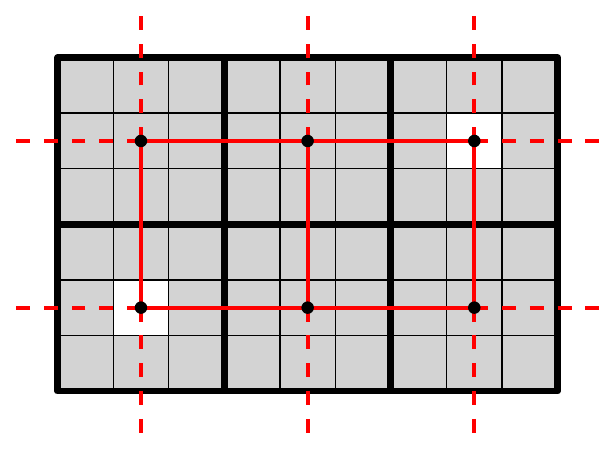}
		}\hfill
		\subfigure[]{
			\includegraphics[page=2,scale=0.7]{ipe/04_connectivity.pdf}
		}
	\caption[Movement Grid]{(a) An already scaled polyomino with a scaling factor of \(c = 3\).
The robot moves along the red lines when an algorithm is executed.
(b) The robot enters an empty vertex when it tries to move southwards.}
	\label{AdaptAlg_gfx_movement}
\end{figure}

\begin{figure}[h]\centering
	\subfigure[]{
		\includegraphics[scale=0.7]{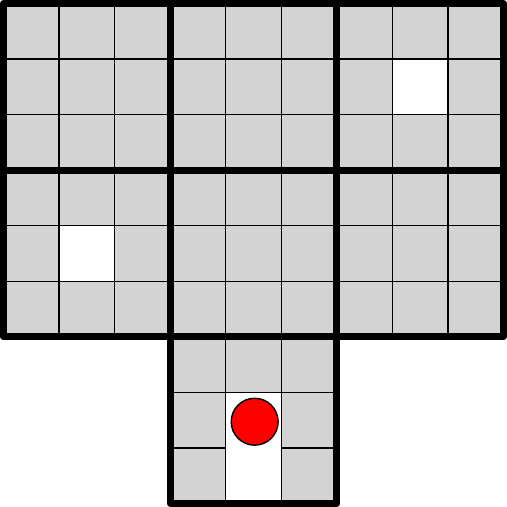}
	}\hfil
	\subfigure[]{
		\includegraphics[scale=0.7]{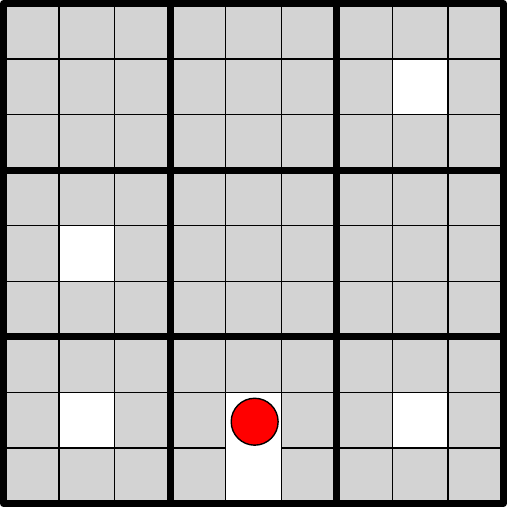}
	}
	\caption[Entering a New Row]{(a) Marked the first \(c \times c\) segment, which the robot entered.
(b) New entered row filled with \(c \times c\) segments and retraced to the previously marked segment.}
	\label{AdaptAlg_gfx_fillRow}
\end{figure}

\section{Scaling Monotone Polyominoes}
\label{monotone}

In case the polyomino to scale is monotone with respect to at least one axis, we can scale it using two robots as well, but without the preceding bounding box construction.
Our scaling strategy is split into two phases.
Firstly we scale columnwise in the direction, in which \(P\) is monotone and afterward, also columnwise, in the other direction.
In the following, we describe the case of scaling an x-monotone polyomino.
The y-monotone case is handled from a rotated view, i.e., we move left instead of up, right instead of down, and so on.
\Cref{monotone_gfx_initial-scaled} shows an example of an x-monotone polyomino (a) and its scaled version (b).

\begin{figure}[h]\centering
		\subfigure[]{\includegraphics[scale=0.7]{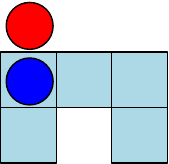}
		}\hfil
		\subfigure[]{
			\includegraphics[scale=0.7]{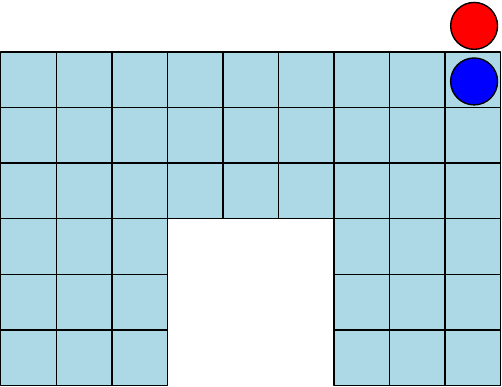}
		}
	\caption[Scaling Monotone Polyominoes]{(a) An x-monotone input polyomino \(P\).
(b) \(P\) scaled by factor \(c=3\).}
	\label{monotone_gfx_initial-scaled}
\end{figure}

Initially, the two robots must be placed together on the polyomino \(P\).
The first robot \(R1\) as the leader, who places and removes tiles and leads the movement.
And the second robot \(R2\) for marking tasks.
We assume monotony and start by moving to the leftmost column and afterward, to the topmost tile.
From that, the actual scaling process can be started with the first phase.

\subsection{Phase I}
During the first phase, we process the polyomino columnwise from left to right resulting in a polyomino scaled in the y-direction by the factor \(c\).
The marker robot \(R2\) is placed on the uppermost tile of the current column and \(R1\) right above.
For every tile in that column \(R2\) waits for \(R1\) to communicate the assignment of the current vertex and moves one step down until the lower end is reached.
After \(R1\) got a signal for an occupied vertex, it extends the current row to the upper side by \(c-1\) tiles and returns to \(R2\).
When \(R2\) steps on the first empty vertex, the current column is finished after \(R1\) returns to \(R2\), see \Cref{monotone_gfx_column-phase1} (b) for an example.

\begin{figure}[h]
	\centering
		\subfigure[]{
			\includegraphics[scale=0.7]{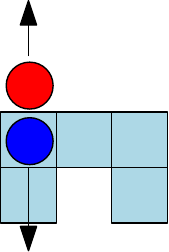}
		}\hfil
		\subfigure[]{
			\includegraphics[scale=0.7]{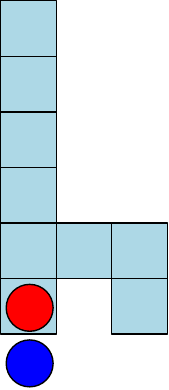}
		}
	\caption[Scaled Column Monotone \(P\)]{(a) Processing one column during the first phase, \(R1\) extends the column to the upper side and \(R2\) marks the current progress within that column.
(b) The first column scaled in y-direction by factor \(c=3\).}
	\label{monotone_gfx_column-phase1}
\end{figure}

After every finished column during Phase I, it is ensured that all columns to the left are scaled in the y-direction.
Moving on to the next column may lead to one of three possible cases depending on the lowest tiles y-coordinate of the previously finished and the following column.
The lowest tile of the next column may have an equal, a higher or lower y-coordinate as the lowest tile of the previously finished column.
That results in different handling.
Let \(t_p\) (\(t_c\)) be the lowest tile within the previously handled (current) column and \(y_{t_p}\) (\(y_{t_c}\)) the y-coordinate of that tile.
Then the three cases are handled as follows:

\begin{enumerate}
	\item \(y_{t_p} = y_{t_c}\): No additional preparation needed, we simply continue as before by scaling the next column in the y-direction.
	\item \(y_{t_p} < y_{t_c}\): To avoid displacement in cause of empty vertices underneath the next column, each of these empty vertices has to be also scaled.
Therefore, the robots start with the first empty vertex right beneath the next column.
From that, for every empty vertex with a higher or equal y-coordinate, every unprocessed column of \(P\), i.e., all columns to the right, will be shifted \(c-1\) steps upwards.
This can be done clockwise by placing \(c-1\) tiles above every column, followed by removing \(c-1\) tiles from the bottom end of every column.
See \Cref{monotone_gfx_higherY} for an example.
	\begin{figure}[h]\centering
		\subfigure[]{
			\includegraphics[scale=0.7]{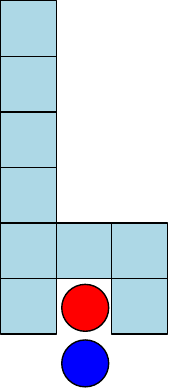}
		}\hfil
		\subfigure[]{
			\includegraphics[scale=0.7]{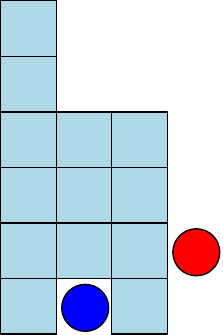}
		}\hfil
		\subfigure[]{
			\includegraphics[scale=0.7]{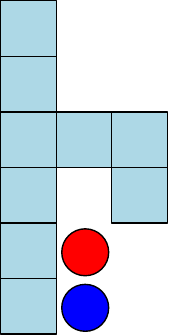}
		}\hfill
		\caption[Shifting Unhandled Columns]{(a) Recognizing \(y_{t_p} < y_{t_c}\), so in this example every unprocessed column needs to be shifted upwards \(c-1\) steps one time.
(b) Two tiles placed above every following column.
(c) Two tiles removed from the bottom end of every column.}
		\label{monotone_gfx_higherY}
	\end{figure}
	\item \(y_{t_p} > y_{t_c}\): In this case, we shift the further unprocessed columns \(c-1\) steps upwards for every occupied vertex with a higher or equal y-coordinate.
The shifting can be done as described in the second case.
Afterward, the scaling can continue reflected horizontally, i.e., for every tile within a column, the robot extends it by placing \(c-1\) tiles right beneath the lower end of the column instead of the upper side.
The possible cases when entering a new column have to be detected mirrored, too.
That means we execute the procedure as described for \(y_{t_p} < y_{t_c}\) when we recognize \(y_{t_p} > y_{t_c}\) and vice versa and concerning the uppermost tile instead of the lowermost.
Moreover, a possibly shifting procedure has to be done counterclockwise instead of clockwise.
	\begin{figure}[h]\centering
			\subfigure[]{
				\includegraphics[scale=0.7]{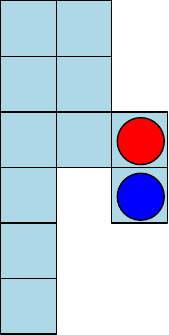}
			}\hfil
			\subfigure[]{
				\includegraphics[scale=0.7]{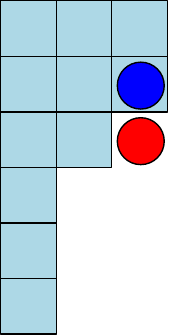}
			}\hfil
			\subfigure[]{
				\includegraphics[scale=0.7]{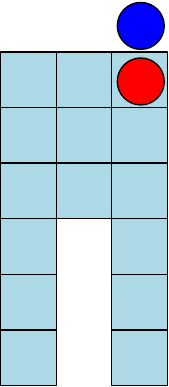}
			}
		\caption[Scaling Horizontally Mirrored]{(a) Recognizing \(y_{t_p} > y_{t_c}\), so in this example every unprocessed column needs to be shifted upwards \(c-1\) steps one time.
And afterward, the column scaling will continue horizontally mirrored.
(b) After shifting is done, the column scaling in y-direction will continue horizontally mirrored.
So the additional tiles get placed below instead of above the column.
(c) Extended the current column downwards.}
		\label{monotone_gfx_lowerY}
	\end{figure}
\end{enumerate}

Using this strategy enables us to scale the polyomino in one direction.
Afterward, the scaling has to be done in the other direction, which we separate in Phase II.

\subsection{Phase II}
During Phase II, every column has to be scaled on the other axis, here in the x-direction.
This can be done by copying every column \(c-1\) times next to its current position.
Except for the first column, copying one column results in a conflict with its neighbor.
Therefore, before we duplicate a column within the polyomino, every column we have processed so far has to be shifted outwards \(c-1\) steps in x-direction.
This can be done by repeating the following steps, starting with \(R1\) on the leftmost column and \(R2\) marking the next column to scale, as shown in \Cref{monotone_gfx_duplicateCol} (a).

\begin{enumerate}
	\item Duplicate the current column to shift \(c-1\) times to the left.
	\item From the uppermost tile of the current column, see \Cref{monotone_gfx_duplicateCol} (b), the robot moves \(c\) steps to the right and searches for the first tile within the entered column by moving down.
	\item After the first tile is found, \(R1\) moves additional \(\lfloor \frac{c}{2} \rfloor \) steps down, followed by moving \(c-1\) steps to the left.
The reached vertex will be used as a marker when removing the previously shifted column.
\Cref{monotone_gfx_duplicateCol} (c) shows the configuration after executing this step.
The robot \(R1\) is placed on top of the marker vertex.
	\item We remove all tiles from above and underneath of this marker.
Firstly done from one end of the column and afterward, from the other end.
This step is executed for \(c-1\) columns.
\(R1\) ends on the next column to shift, or in case the second robot \(R2\) can be found at the bottom end of that column, shifting is done.
Now, the next column to scale can be duplicated.
\end{enumerate}

	\begin{figure}[h]\centering 
		\hfil\subfigure[]{\includegraphics[scale=0.6]{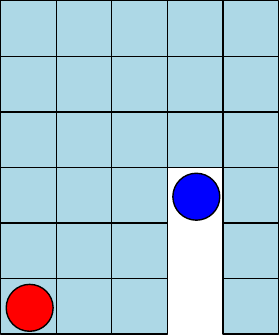}
		}\hfil
		\subfigure[]{\includegraphics[scale=0.6]{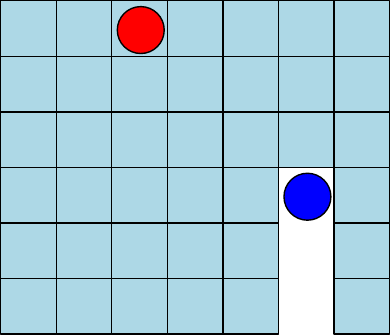}
		}\hfil
		\subfigure[]{\includegraphics[scale=0.6]{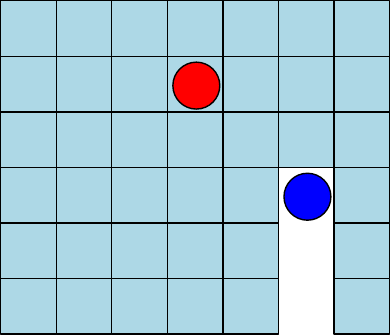}
		}
	
		\hfil\subfigure[]{\includegraphics[scale=0.6]{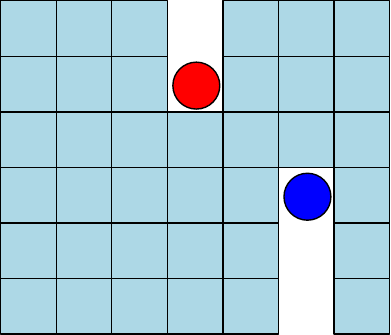}
		}\hfil
		\subfigure[]{\includegraphics[scale=0.6]{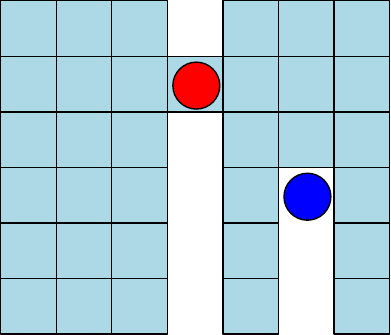}
		}\hfil
		\subfigure[]{\includegraphics[scale=0.6]{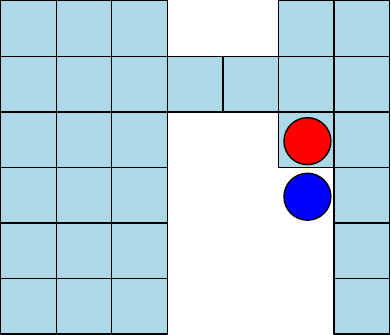}
		}
	\caption[Duplicate Column]{(a) The first column is duplicated \(c-1\) times.
The next column to scale is marked by \(R2\) staying beneath it.
\(R1\) starts the shifting procedure with the leftmost already processed column.
(b) The current column to shift is duplicated \(c-1\) times.
(c) Found the marker tile after executing Step 3.
(d) Removed all tiles above the marker.
(e) Removed all tiles underneath the marker.
(f) Shifting one column by \(c-1\) steps is done.}
	\label{monotone_gfx_duplicateCol}
	\end{figure}

\begin{theorem}\label{th:scale_monotone}
	Scaling an x-monotone polyomino \(P\) of width \(w\) and height \(h\) by a constant scaling factor \(c\) without loss of connectivity can be done with two robots in \(O(w^2hc^2)\) steps.
\end{theorem}

\begin{proof}
	\textbf{Correctness:} During Phase I the polyomino is processed columnwise from left to right.
Using the second robot as a marker within every column, we can track the progress within it, and we only move on to the next column, when every tile of the current column is scaled in the y-direction.
In case we have to shift tiles to represent the scaling of an empty vertex, the whole part of the polyomino which we have not processed so far will be shifted to ensure connectivity between the columns.
This shifting is done clockwise (or in mirrored case, counterclockwise), i.e., firstly we extend all columns by \(c-1\) tiles on one side, and secondly, we reduce all columns by \(c-1\) tiles on the other side.
This method ensures on the one hand that all unprocessed columns so far are properly aligned and on the other that \(P\) stays connected during the whole procedure.
This first phase terminates when the rightmost column is finished.
	In Phase II we have to scale every column in the x-direction, so we also move on columnwise.
Whenever it is needed to shift the already processed columns to free space for the next column to scale, one robot stays on the next column to scale, while the other does the actual shifting.
We use bridges to ensure connectivity between two adjacent columns when shifting one of them.
These bridges are placed between the first found connection between them, i.e., the first row where a tile is placed in both columns when the robot moves top-down.
Shifting is finished, when the next column to scale is reached, which is marked by the second robot.
When there is no further column to the right, the algorithm terminates with \(P\) scaled by the constant \mbox{scaling factor \(c\)}.
	\\
	\textbf{Time:} Within every column of \(P\) there are at most \(O(h)\) tiles, so scaling every tile in y-direction takes \(O(hc)\) steps per column and \(O(whc)\) steps in total.
Shifting every unprocessed column takes \(O(wc)\) steps and may be needed up to \(O(h)\) times per column.
This yields in a total running time of \(O(whc + w^2hc) \subset O(w^2hc)\) for the first phase.
	During Phase II, for every column to scale, we have to shift all previously handled columns \(c-1\) steps outwards, what results in at most \(O(whc^2)\) steps per column.
Scaling one column during Phase II needs \(O(hc^2)\) steps.
Therefore, we get a total of \(O(w^2hc^2)\) steps for scaling an x-monotone polyomino.
\end{proof}

\section{Conclusion}
\subsection{Presented Work}
The "Robot-on-Tiles" model provides an infinite grid graph with robots acting like a deterministic finite automaton and with the ability to manipulate polyominoes that are placed on that grid.
This model is well designed to investigate problems related to programmable matter.
We have shown that an arbitrary polyomino can be scaled by a constant factor \(c\) without loss of connectivity, using up to two robots, depending on the input polyomino.
We presented different approaches, firstly for polyominoes in general and with slight modifications for simple polyominoes, and secondly for polyominoes that are monotone with respect to at least one axis. 
With knowledge about the ability to scale polyominoes without loss of connectivity, we can adapt most algorithms, which perform arbitrary tasks on polyominoes within the same model and possibly not ensure connectivity, with an additional account for the requirement for connectivity.

In \Cref{BB} we presented an algorithm to construct a bounding box surrounding a given polyomino of width \(w\) and height \(h\) using two robots in a running time of \(O(max(w,h) \cdot (wh + k \cdot |\partial P|))\) steps, where \(k\) is the number of convex corners.
One robot holds the bounding box and the polyomino together, while the other robot performs the actual construction.
The construction is performed clockwise around the polyomino, and whenever the robot comes in conflict with an already existing tile, we use a subroutine to check whether that tile belongs to the already constructed bounding box or the initial polyomino.
With slight modifications that were described in \Cref{BB_Simple}, we can construct the bounding box for a simple polyomino with the use of just one robot and a running time of \(O(max(w,h) \cdot wh)\) steps.

Using the previously described bounding box construction, we introduced our scaling algorithm in \Cref{Scale}.
Once the bounding box is constructed, the initial polyomino can be scaled right next to the bounding box using one robot.
For every vertex within the bounded area, the robot assembles a \(c \times c\) segment that represents either an occupied or an empty vertex.
Hence we can distinguish these types of vertices in the area of the scaled polyomino.
Our scaling algorithm has a total running time of \(O(wh \cdot (c^2 + cw + ch))\) steps using two robots.
Since the use of two robots is only reasoned by the previously bounding box construction, we can scale a simple polyomino using one robot within the same time complexity.
Given a polyomino that is already scaled by a factor \(c\), a very similar strategy for down-scaling by the inverse factor \(\frac{1}{c}\) was presented in \Cref{DownScaling}.

In case the polyomino to scale is monotone, we can use another scaling strategy, which is described in \Cref{monotone} and gets along without a previously constructed bounding box.
We use two robots and start by extending every column by \(c-1\) tiles for every initial tile within that column, followed by copying every column \(c-1\) times.
We end in a running time of \(O(w^2hc^2)\) steps using two robots.

Finally, we introduced a strategy to adapt most algorithms, which can be performed on a polyomino within the robots on tiles model and not necessarily ensure connectivity, with the additional requirement of connectivity in \Cref{AdaptAlg}.
We use our described scaling strategy with minor adaptions to receive a tiled rectangle containing the scaled polyomino.
Afterward, we maintain this rectangle by filling up every empty column or row which we entered while executing the adapted algorithm.
Considering the running time for previously scaling, and for maintaining the received rectangle, our strategy results in a total running time of \(O(wh\cdot (c^2 + cw + ch) + max((w'-w)h',(h'-h)w') + c\cdot \mathcal{T(A)})\) steps.

\subsection{Future Work}
There is a whole range of possible extensions.
 
Is it possible to scale general polyominoes without the preceding bounding box
construction?  A possible approach could be to cut the polyomino into a subset
of monotone polyominoes, which could be handled separately.

Another challenge is to develop distributed algorithms with multiple robots that are
capable of solving a range of problems with the requirement of connectivity,
without having to rely on the 
preceding scaling procedure that we used in our work.  
Other questions arise from additional requirements of real-world applications
e.g., in the construction and reconfiguration of space habitats.

\bibliographystyle{abbrv}
\bibliography{bibliography}

\end{document}